\documentclass{CSML}
\pdfoutput=1

\usepackage{lastpage}
\newcommand{\dOi}{14(1:25)2018}
\lmcsheading{}{1--\pageref{LastPage}}{}{}%
{Feb.~24, 2014}{Mar.~30, 2018}{}

\usepackage{stmaryrd,bcprulesmhfix,tensor}
\bcprulessavespacetrue
\suppressrulenamesfalse
\usepackage{mathpartir}
\usepackage{upgreek}
\usepackage{proof}
\usepackage[all]{xy}
\usepackage{graphicx}
\usepackage{url,amssymb}
\renewcommand{\phi}{\varphi}
\usepackage{xspace}
\RequirePackage{txfonts}
\newif\iffull\fulltrue

\usepackage{hyperref}
\hypersetup{hidelinks}

\newcommand{\vdashv}{\vdash_v}
\newcommand{\vdashc}{\vdash_c}
\newcommand{\keywd}[1]{\mathtt{#1}}

\newcommand{\wcpo}[1]{$\omega$-cpo}
\newcommand{\wcpos}[1]{$\omega$-cpos}

\newcommand{\sq}[4]{\tensor*[^{#1}_{#2}]{\Diamond}{^{#3}_{#4}}}

\newcommand{\spn}[5]{{#1}\stackrel{#2}{\leftarrow}{#3}
         \stackrel{#4}{\rightarrow}{#5}}
\newcommand{\lspn}[5]{{#1}\stackrel{#2}{\longleftarrow}{#3}
         \stackrel{#4}{\longrightarrow}{#5}}

\newcommand{\supp}{\operatorname{supp}}

\newcommand{\mtrue}{\keywd{true}}
\newcommand{\mfalse}{\keywd{false}}

\newcommand{\ints}{\mathbb{Z}}
\newcommand{\bools}{\mathbb{B}}

\newcommand{\nats}{\mathbb{N}}
\newcommand{\inttype}{\keywd{int}}

\newcommand{\booltype}{\keywd{bool}}
\newcommand{\loctype}{\keywd{name}}

\newcommand{\dom}[1]{\mathrm{dom}({#1})}
\newcommand{\cod}[1]{\mathrm{cod}({#1})}

\newcommand{\pause}{\vspace{1.5ex}}
\newcommand{\gap}{\quad\quad}

\newcommand\Pscr{\mathcal{P}}
\newcommand\Iscr{\mathcal{I}}

\newcommand{\ba}{\begin{array}}
\newcommand{\ea}{\end{array}}

\newcommand{\squelch}[1]{}

\newcommand{\New}{\texttt{new}}

\newcommand{\vfix}[3]{\keywd{rec}\:{#1}\:{#2} = {#3}}
\newcommand{\vfun}[2]{\keywd{fun}\:{#1}.{#2}}
\newcommand{\letin}[2]{\keywd{let}\:{#1}\!\Leftarrow\!{#2}\:\keywd{in}\:}

\newcommand{\opletrec}[3]{\keywd{let\ rec}}

\newcommand{\myif}[3]{\keywd{if}\ #1\ \keywd{then}\ #2\
  \keywd{else}\ #3}
\newcommand{\lett}{{\keywd{let}\ }}

\newcommand{\fullonly}[1]{}

\newlength{\lruleZeroname}
\newcommand{\domL}[1]{\ensuremath{\mathrm{dom}}(#1)}
\newcommand{\domR}[1]{\ensuremath{\mathrm{dom}'}(#1)}
\newcommand{\lop}[1]{\ensuremath{\mathrm{lop}}(#1)}

\newcommand{\sem}[1]{\ensuremath{\llbracket {#1} \rrbracket}}

\newcommand{\cloc}{\ensuremath{\mathsf{l}}\xspace}

\newcommand\w{\ensuremath{\mathsf{w}}\xspace}
\newcommand\q{\ensuremath{\mathsf{q}}\xspace}

\newcommand{\world}{\ensuremath{\mathbf{W}}\xspace}

\newcommand{\semV}[1]{\ensuremath{\llceil {#1} \rrceil
}}

\newcommand\val{\ensuremath{\mathsf{v}}\xspace}
\newcommand\vval{\ensuremath{v}\xspace}
\newcommand\bval{\ensuremath{\mathsf{b}}\xspace}
\newcommand\bbval{\ensuremath{b}\xspace}
\newcommand\ival{\ensuremath{\mathsf{i}}\xspace}
\newcommand\iival{\ensuremath{i}\xspace}

\newcommand\cval{\ensuremath{\mathsf{c}}\xspace}

\newcommand\ccval{\ensuremath{c}\xspace}

\newcommand\Cscr{\ensuremath{\mathcal{C}}\xspace}
\newcommand\Std{\ensuremath{\mathit{Std}}\xspace}

\newcommand\ie{\emph{i.e.}\xspace}
\newcommand\eg{\emph{e.g.}\xspace}

\newcommand\id{id}
\newcommand{\svf}{SVF}
\newcommand{\svfs}{SVFs}

\newcommand{\todo}[1]{\textcolor{red}{#1} }
\begin{document}
%
%
\title[Proof-Relevant Logical Relations for Name Generation]{Proof-Relevant Logical Relations for Name Generation}
\author[Nick Benton et al.]{Nick Benton}	
\address{Facebook, London, UK}	
\email{nick.benton@gmail.com}  

\author[]{Martin Hofmann}	
\address{LMU, Munich, Germany}	
\email{hofmann@ifi.lmu.de}  

\author[]{Vivek Nigam}	
\address{UFPB, Jo\~ao Pessoa, Brazil \& fortiss, Munich, Germany}	
\email{vivek.nigam@gmail.com}  



\titlecomment{A preliminary version of this work was presented at the 
11th International Conference on Typed Lambda Calculi and Applications, TLCA 2013, 
26-28 June 2013, Eindhoven, The Netherlands.}

\begin{abstract}
Pitts and Stark's $\nu$-calculus is a paradigmatic total language for
studying the problem of contextual equivalence in higher-order
languages with name generation. Models for the $\nu$-calculus that
validate basic equivalences concerning names may be constructed using
functor categories or nominal sets, with a dynamic allocation monad
used to model computations that may allocate fresh names. If recursion
is added to the language and one attempts to adapt the models from
(nominal) sets to (nominal) domains, however, the direct-style
construction of the allocation monad no longer works. This issue has
previously been addressed by using a monad that combines dynamic
allocation with continuations, at some cost to abstraction.

This paper presents a direct-style model of a $\nu$-calculus-like
language with recursion using the novel framework of \emph{proof-relevant
logical relations}, in which logical relations also contain objects (or
proofs) demonstrating the equivalence of (the semantic counterparts
of) programs. Apart from providing a fresh solution to an old problem,
this work provides an accessible setting in which to introduce the use of
proof-relevant logical relations, free of the additional complexities
associated with their use for more sophisticated languages.
\end{abstract}

\maketitle


\section*{Introduction}
\label{sec:intro}
Reasoning about contextual equivalence in higher-order languages that
feature dynamic allocation of names, references, objects or keys is
challenging. Pitts and Stark's $\nu$-calculus boils the problem down
to its purest form, being a total, simply-typed lambda calculus with
just names and booleans as base types, an operation $\New$ that generates fresh
names, and equality testing on names. The full equational theory of the
$\nu$-calculus is surprisingly complex and has been studied both
operationally and denotationally, using logical relations
\cite{Stark94thesis,pitts98high},
environmental bisimulations \cite{bentonkoutavas} and nominal game
semantics \cite{AbramskyEtal04lics,tzevelekos}.

Even before one considers the `exotic' equivalences that arise from
the (partial) encapsulation of names within closures, there are two
basic equivalences that hold for essentially all forms of
generativity:
\[
 \begin{array}{ll}
  (\letin{x}{\New} e) = e, \textrm{ provided $x$ is not free in $e$.}
   & \mbox{(Drop)}\\
  (\letin{x}{\New} \letin{y}{\New} e) = (\letin{y}{\New}
   \letin{x}{\New} e) & \mbox{(Swap)}.
 \end{array}
\]
The (Drop) equivalence says that removing the generation of unused
names preserves behaviour; this is sometimes called the `garbage
collection' rule. The (Swap) equivalence says that the order in which
names are generated is immaterial. These two equations also appear as
structural congruences for name restriction in the $\pi$-calculus.

Denotational models for the $\nu$-calculus validating (Drop) and
(Swap) may be constructed using (pullback-preserving) functors in
$\mathit{Set}^{\world}$, where $\world$ is the category of finite sets
and injections \cite{Stark94thesis}, or in
FM-sets~\cite{DBLP:journals/fac/GabbayP02}.  These models use a
dynamic allocation monad to interpret possibly-allocating
computations. One might expect that moving to $\mathit{Cpo}^{\world}$
or FM-cpos would allow such models to adapt straightforwardly to a
language with recursion, and indeed Shinwell, Pitts and Gabbay
originally proposed~\cite{shinwell03icfp} a dynamic allocation monad
over FM-cpos. However, it turned out that the underlying FM-cppo of
the proposed monad does not actually have least upper bounds for all
finitely-supported chains. A counter-example is given in Shinwell's
thesis~\cite[page 86]{shinwell04phd}.  To avoid the problem, Shinwell
and Pitts \cite{shinwell04phd,shinwell05tcs} moved to an
\emph{indirect-style} model, using a \emph{continuation
  monad}~\cite{pitts98high}: $(-)^{\top \top} \stackrel{def}{=} (-
\rightarrow 1_\bot) \rightarrow 1_\bot$ to interpret computations.  In
particular, one shows that two programs are equivalent by proving that
they co-terminate when supplied with the same (or equivalent)
continuations. The CPS approach was also adopted by Benton and
Leperchey \cite{DBLP:conf/tlca/BentonL05}, and by Bohr and Birkedal
\cite{bohrbirkedal}, for modelling languages with references.

In the context of our on-going research on the semantics of
effect-based program transformations
\cite{DBLP:conf/aplas/BentonKHB06}, we have been led to develop
\emph{proof-relevant} logical relations~\cite{benton14popl}. These
interpret types not merely as partial equivalence relations, as is
commonly done, but as a proof-relevant generalization thereof:
\emph{setoids}. A setoid is like a category all of whose morphisms are
isomorphisms (a groupoid) with the difference that no equations
between these morphisms are imposed. The objects of a setoid establish
that values inhabit semantic types, whilst its morphisms are
understood as explicit proofs of semantic equivalence. This paper
shows how we can use proof-relevant logical relations to give a
direct-style model of a language with name generation and recursion,
validating (Drop) and (Swap). Apart from providing a fresh approach to
an old problem, our aim in doing this is to provide a comparatively
accessible presentation of proof-relevant logical relations in a
simple setting, free of the extra complexities associated with
specialising them to abstract regions and effects \cite{benton14popl}. 

Although our model validates the two most basic equations for name
generation, it is -- like simple functor categories in the total case
-- still far from fully abstract. Many of the subtler contextual
equivalences of the $\nu$-calculus still hold in the presence of
recursion; one naturally wonders whether the more sophisticated
methods used to prove those equivalences carry over to the
proof-relevant setting. We will show one such method, Stark's
\emph{parametric functors}, which are a categorical version of Kripke
logical relations, does indeed generalize smoothly, and can be used to
establish a non-trivial equivalence involving encapsulation of fresh
names. Moreover, the proof-relevant version is naturally transitive,
which is, somewhat notoriously, not generally true of ordinary logical
relations.

Section~\ref{sec:lang} sketches the language with which we will be
working, and a naive `raw' domain-theoretic semantics for it. This
semantics does not validate interesting equivalences, but is
adequate. By constructing a realizability relation between it and the
more abstract semantics we subsequently introduce, we will be able to
show adequacy of the more abstract semantics.  In
Section~\ref{sec:setoids} we introduce our category of setoids; these
are predomains where there is a (possibly-empty) set of `proofs'
witnessing the equality of each pair of
elements. We then describe pullback-preserving
functors from the category of worlds $\world$ into the category of
setoids. Such functors will interpret types of our language in the
more abstract semantics, with morphisms between them interpreting
terms. The interesting construction here is that of a dynamic
allocation monad over the category of pullback-preserving
functors. Section~\ref{sec:prlr} shows how the abstract semantics is
defined and related to the more concrete
one. Section~\ref{sec:examples} then shows how the semantics may be
used to establish basic equivalences involving name generation. Section~\ref{sec:parametric} describes how proof-relevant parametric functors can validate a more subtle equivalence involving encapsulation of new names.

\section{Syntax and Semantics}
\label{sec:semantics}
\label{sec:lang}
We work with an entirely conventional CBV language, featuring recursive functions and base types that include names, equipped with equality testing and fresh name generation (here $+$ is just a representative operation on integers):
\begin{eqnarray*}
\tau & := & \inttype \mid \booltype \mid \loctype \mid \tau\to\tau'\\[2pt]
v & := & x \mid b \mid i \mid \vfix{f}{x}{e} \mid v + v' \mid v = v'\\[2pt]
e & := & v  \mid \New \mid \letin{x}{e}{e'} \mid v\,v'
 \mid \myif{v}{e}{e'}\\[2pt]
\Gamma & := & x_1:\tau_1,\ldots,x_n:\tau_n
\end{eqnarray*}
The expression $\vfix{f}{x}{e}$ stands for an anonymous function which satisfies the recursive equation $f(x)=e$ where both $x$ and $f$ may occur in $e$. In the special case where $f$ does not occur in $e$, the construct degenerates to function abstraction. We thus introduce the abbreviation: 
\[
\vfun{x}{e} \triangleq \vfix{f}{x}{e}\quad\mbox{where $f$ does not occur in $e$}.
\]
There are typing judgements for values,
$\Gamma \vdashv v:\tau$,
and computations, $\Gamma \vdashc e:\tau$,
defined in an unsurprising way; these are shown in Figure~\ref{fig:typing}. We will often elide the subscript on turnstiles.
\begin{figure*}[tp]
\label{fig:typing}
\begin{mathpar}
\infer{\Gamma,x:\tau\vdashv x:\tau}{}
\and
\infer{\Gamma\vdashv b:\booltype}{}
\and
\infer{\Gamma\vdashv i:\inttype}{}
\and
\inferrule{\Gamma,f:\tau\to\tau',x:\tau\vdashc e:\tau'}{\Gamma\vdashv\vfix{f}{x}{e}:\tau\to\tau'}
\\
\inferrule{\Gamma\vdashv v:\inttype\\ \Gamma\vdashv v':\inttype}{\Gamma\vdashv v+v' : \inttype}
\and
\inferrule{\Gamma\vdashv v:\tau \\ \Gamma\vdashv v':\tau\\ \tau\in\{\inttype,\loctype\}}{\Gamma\vdashv v=v' :\booltype}
\and
\inferrule{\Gamma\vdashv v:\tau}{\Gamma\vdashc v:\tau}
\and
\infer{\Gamma\vdashc \New : \loctype}{}
\and
\inferrule{\Gamma\vdashc e:\tau\\ \Gamma,x:\tau\vdashc e':\tau'}{\Gamma\vdashc \letin{x}{e}{e'} : \tau'}
\and
\inferrule{\Gamma\vdashv v: \tau\to\tau'\\ \Gamma\vdashv v':\tau}{\Gamma\vdashc v\,v' : \tau'}
\and
\inferrule{\Gamma\vdashv v:\booltype\\ \Gamma\vdashc e:\tau\\ \Gamma\vdashc e':\tau}{\Gamma\vdashc \myif{v}{e}{e'} : \tau}
\end{mathpar}
\caption{Typing rules for language with recursion and name generation.}
\end{figure*}

We define a simple-minded concrete denotational semantics $\semV{\cdot}$ for this language using predomains ($\omega$-cpos) and continuous maps. For types we take
\[\begin{array}{l}
\semV{\inttype}=\ints\qquad 
\semV{\booltype} = \bools \qquad
\semV{\loctype} = \nats\\[4pt]
\semV{\tau\to\tau'} = \semV{\tau}\to (\nats\to \nats\times\semV{\tau'})_\bot\\[4pt]
\semV{x_1:\tau_1,\ldots,x_n:\tau_n} = \semV{\tau_1}\times\cdots\times\semV{\tau_n}
\end{array}
\]
and there are then conventional clauses defining
\[
\begin{array}{l}
\semV{\Gamma\vdashv v:\tau} : \semV{\Gamma}\to\semV{\tau}\qquad\mbox{and}\qquad
\semV{\Gamma\vdashc e:\tau} : \semV{\Gamma}\to (\nats\to\nats\times\semV{\tau})_\bot
\end{array}
\]
Note that this semantics just uses naturals to interpret names, and a state monad over names to interpret possibly-allocating computations. For allocation we take
\[\semV{\Gamma\vdashc\New:\loctype}(\eta) = [\lambda n. (n+1,n)]\] 
returning the next free name and incrementing the name supply. This semantics validates no interesting equivalences involving names, but is adequate for the obvious operational semantics. Our more abstract semantics, $\sem{\cdot}$, will be related to $\semV{\cdot}$ in order to establish \emph{its} adequacy.

\section{Setoids}
\label{sec:setoids}
We define the \emph{category of setoids}, $\Std$, to be the exact
completion of the category of predomains, see
\cite{DBLP:conf/mfps/CarboniFS87,DBLP:conf/lics/BirkedalCRS98}. We
give here an elementary description of this category using the
language of dependent types.  A \emph{setoid} $A$ consists of a
predomain $|A|$ and for any two $x,y\in |A|$ a set $A(x,y)$ of
``proofs'' (that $x$ and $y$ are equal). The set of triples $X = \{(x,y,p)
\mid p\in A(x,y)\}$ must itself be a predomain, \ie, there has to be an order relation $\leq$ such that $(X,\leq)$ is a predomain. The first and
second projections out of the set of triples  must be continuous. Furthermore, there are
continuous functions $r_A:\Pi x\in |A|.A(x,x)$ and $s_A:\Pi x,y\in
|A|.A(x,y)\rightarrow A(y,x)$ and $t_A:\Pi x,y,z.A(x,y)\times
A(y,z)\rightarrow A(x,z)$, witnessing reflexivity, symmetry and
transitivity; note that, unlike the case of \emph{groupoids}, no equations involving $r$, $s$ and $t$ are imposed.

We should explain what continuity of a dependent function like
$t(-,-)$ is: if $(x_i)_i$ and $(y_i)_i$ and $(z_i)_i$ are ascending
chains in $A$ with suprema $x,y,z$ and $p_i\in A(x_i,y_i)$ and $q_i\in
A(y_i,z_i)$ are proofs such that $(x_i,y_i,p_i)_i$ and $(y_i,z_i,q_i)_i$
are ascending chains, too, with suprema $(x,y,p)$ and $(y,z,q)$ then
$(x_i,z_i,t(p_i,q_i))$ is an ascending chain of proofs (by
monotonicity of $t(-,-)$) and its supremum is $(x,z,t(p,q))$.
Formally, such dependent functions can be reduced to non-dependent ones using
pullbacks, that is $t$ would be a function defined on the pullback of
the second and first projections from $\{(x,y,p)\mid p\in A(x,y)\}$ to
$|A|$, but we find the dependent notation to be much more readable.
If $p\in A(x,y)$ we may write $p:x\sim y$ or simply $x\sim y$. We also
omit $|-|$ wherever appropriate. We remark that ``setoids'' also appear in
constructive mathematics and formal proof, see \eg,
\cite{DBLP:journals/jfp/BartheCP03}, but the proof-relevant nature of
equality proofs is not exploited there and everything is based on sets
(types) rather than predomains. 
A morphism from setoid $A$ to setoid $B$ is an equivalence class of  pairs $f=(f_0,f_1)$ of
continuous functions where $f_0:|A|\rightarrow |B|$ and $f_1:\Pi
x,y\in|A|.A(x,y)\rightarrow B(f_0(x),f_0(y))$. Two such pairs
$f,g:A\rightarrow B$ are \emph{identified} if there exists a
continuous function $\mu:\Pi a\in|A|.B(f_0(a),g_0(a))$. 

The following is folklore, see also \cite{DBLP:conf/lics/BirkedalCRS98}. 
\begin{prop}
The category of setoids is cartesian closed. Cartesian product is given pointwise.  The function space $A\Rightarrow B$ 
of setoids $A$ and $B$ is given as follows: the underlying predomain $|A\Rightarrow B|$ comprises pairs $(f_0,f_1)$ which are \emph{representatives} of morphisms from $A$ to $B$. That is, $f_0:|A|\rightarrow |B|$ and $f_1:\Pi
x,y\in|A|.A(x,y)\rightarrow B(f_0(x),f_0(y))$ are continuous functions with the pointwise ordering. The proof set $(A\Rightarrow B)((f_0,f_1), (f_0',f_1'))$ comprises witnesses of the equality of $(f_0,f_1)$ and $(f_0',f_1')$ qua morphisms, \ie, continuous functions $\mu:\Pi a\in|A|.B(f_0(a),f_0'(a))$. 
\end{prop}
\begin{proof}
The evaluation morphism $(A\Rightarrow B)\times A\longrightarrow B$
sends $(f_0,f_1)$ and $a$ to $f_0(a)$. If $h:C\times A\longrightarrow
B$ is a morphism represented by $(h_0,h_1)$ then the morphism
$\lambda(h):C\longrightarrow A\Rightarrow B$ may be represented by
$(\lambda(h)_0,\lambda(h)_1)$ where $\lambda(h)_0(c)=(f_0,f_1)$ and
$f_0(a)=h_0(c,a)$ and
$f_1(a,a',p)=h_1((c,a),(c,a'),(r(c),p))$. Likewise,
$\lambda(h)_1(c,c',p)=\mu$ where $\mu(a)=h_1((c,a),(c',a),(p,r(a)))$.
The remaining verifications are left to the reader.
\end{proof}
\begin{defi}
A setoid $D$ is \emph{pointed} if $|D|$ has a least element $\bot$ and
such that there is also a least proof $\bot\in D(\bot,\bot)$. 
If $D$ is pointed we write $\bot$ for the obvious global element
$1\rightarrow D$ returning $\bot$. A morphism $f:D\rightarrow D'$ with
$D,D'$ both pointed is \emph{strict} if $f\bot=\bot$.
\end{defi}
\begin{thm}\label{tfixit}
Let $D$ be a pointed setoid.
Then there
is a morphism of setoids $Y:[D\Rightarrow D]\rightarrow D$
satisfying the following equations (written using $\lambda$-calculus
notation, which is meaningful in cartesian closed categories).
\[
\begin{array}{rclr}
f(Y(f)) &=& Y(f) & \mbox{(Fixpoint)}\\
f(Y(g\circ f)) &=& Y(f\circ g)& \mbox{(Dinaturality)}\\
f(Y(g)) &=& Y(h)\mbox{ if $f$ is strict and $fg=hf$} & \mbox{(Uniformity)}\\
Y(f^n) &=& Y(f) & \mbox{(Power)}\\
Y(\lambda x.f(x,x)) &=& Y(\lambda x.Y(\lambda y.f(x,y)))& \mbox{(Diagonal)}\\
Y(\lambda \vec x.\vec t(\vec x)) &=& \langle Y(s),\dots, Y(s)\rangle & \mbox{(Amalgamation)}\\
\multicolumn{3}{l}{\mbox{when }t_i(y,\dots,y)=s(y)$ for $i=1,\dots,n$ and $\vec t=\langle t_1,\dots,t_n\rangle}
\end{array} 
\]
\end{thm}
\begin{proof}
To define the morphism $Y$ suppose we are given $f=(f_0,f_1)\in |D\Rightarrow D|$. For each $i\in \mathbb{N}$ we define $d_i\in |D|$ 
  by $d_0=\bot$ and $d_{i+1}=f_0(d_i)$. We then put $Y(f)=\sup_i d_i$. 

Now suppose that $f'=(f_0',f_1')\in |D\Rightarrow D|$ and $q:f\sim
f'$, i.e., $q:\Pi d.D(f_0(d),f_0'(d))$. Let $d_i'$ be defined
analogously to $d_i$ so that $Y(f')=\sup_i d_i'$. By induction on $i$
we define proofs $p_i:d_i\sim d_i'$. We put $p_0=\bot$ (the least
proof) and, inductively, $p_{i+1}=t(f_1(p_i),q(d_i'))$
(transitivity). Notice that $f_1(p_i):d_{i+1}\sim f_0(d_i')$ and
$q(d_i'):f_0(d_i')\sim d_{i+1}'$. Now let $(d,d',p)$ be the supremum of the chain 
$(d_i,d_i',p_i)$. By continuity of the projections we have that $d=Y(f)$ and $d'=Y(f')$ and thus $p:Y(f)\sim Y(f')$. The passage from $q$ to $p$ witnesses that $Y$ is indeed a (representative of a)  morphism.

Equations ``Diagonal'' and ``Dinaturality'' follow directly from the
validity of these properties for the least fixpoint combinator for
cpos. For the sake of completeness we prove the second one. Assume
$f,g\in |D\Rightarrow D|$ and let $d_i=(f_0g_0)^i(\bot)$ and
$e_i=(g_0f_0)^i(\bot)$. We have $d_i\leq f_0(e_i)$ and $f_0(e_{i}))\leq
  d_{i+1}$. It follows that $Y(fg)$ and $Y(gf)$ are actually
  equal. Equation ``Fixpoint'' is a direct consequence of dinaturality
  (take $g=\textrm{id}$).

Amalgamation and uniformity are also valid for the least fixpoint
combinator, but cannot be directly inherited since the equational
premises only holds up to $\sim$. As a representative example we show
amalgamation. So assume elements $t_i\in|D^n\Rightarrow D|$ and
$s\in|D\Rightarrow D|$ and proofs $p_k:\Pi
d.D((t_k)_0(d,\dots,d),s(d))$. Consider $d_i=\vec
t_0^i(\bot,\dots,\bot)$ and $e_i=s_0^i(\bot)$. By induction on $i$ and
using the $p_k$ we construct proofs $d_i\sim (e_i,\dots,e_i)$. The
desired proof of $Y(\vec{t})\sim (Y(s),\dots, Y(s))$ is obtained as
the supremum of these proofs as in the definition of the witness that
$Y$ is a morphism above.

Equation ``Power'', finally, can be deduced from amalgamation and dinaturality or alternatively inherited directly from the least fixpoint combinator. 
\end{proof}
The above equational axioms for the fixpoint combinator are taken from Simpson and Plotkin
\cite{DBLP:conf/lics/SimpsonP00}, who show that they imply certain
completeness properties. In particular, it follows that the category
of setoids is an ``iteration theory'' in the sense of Bloom and \'{E}sik
\cite{DBLP:series/eatcs/BloomE93}. For us they are important since the
category of setoids is not cpo-enriched in any reasonable way, so that
the usual order-theoretic characterisation of $Y$ is not
available. Concretely, the equations help, for example, to justify
various loop optimisations when loops are expressed using the fixpoint
combinator. 
\begin{defi}
A setoid $D$ is \emph{discrete} if for all $x,y\in D$ we have $|D(x,y)|\leq 1$ and $|D(x,y)|=1\iff x=y$. 
\end{defi}
Thus, in a discrete setoid proof-relevant equality and actual equality coincide and moreover any two equality proofs are actually equal (i.e.~we have proof irrelevance). 

\section{Finite sets and injections}
\label{sec:pullback}
\emph{Pullback squares} are a central notion in our framework. As it will become
clear later, they are the ``proof-relevant'' component of logical relations.
Recall that a morphism $u$ in a category is a monomorphism if $ux=ux'$
implies $x=x'$ for all morphisms $x,x'$. Two morphisms with common co-domain are called a co-span and two morphisms with common domain are called span. A commuting square $xu=x'u'$
of morphisms is a \emph{pullback} if whenever $xv=x'v'$ there is unique $t$
such that $v=ut$ and $v'=u't$. This can be visualized as follows:
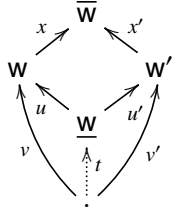
\begin{displaymath}
\vcenter{\xymatrix@C=1pc@R=0.5pc{ 
& \overline{\w}\\
\w\ar[ru]^x & & \w'\ar[lu]_{x'}\\
& \underline{\w}\ar[ru]_{u'} \ar[lu]^u\\
\\
& \cdot\ar@/^/[luuu]^v \ar@/_/[ruuu]_{v'} \ar@{.>}[uu]_t
}}
\end{displaymath}
We write $\sq{x}{u}{x'}{u'}$ or
$\w\sq{x}{u}{x'}{u'}\w'$ (when $\w^( {}' {}^)=\dom{x^( {}' {}^)}$) for such a
pullback square. We call the common codomain of $x$ and $x'$ the
\emph{apex} of the pullback, written $\overline{\w}$, while
the common domain of $u,u'$ is the \emph{low point} of the
square, written $\underline{\w}$. A pullback square $\w\sq{x}{u}{x'}{u'}\w'$ with apex
$\overline{\w}$ is \emph{minimal} if whenever there is another
pullback $\w\sq{x_1}{u}{x_1'}{u'}\w'$ over the same span and with apex
$\overline{\w_1}$, then there is a unique morphism $t:\overline{\w}
\to \overline{\w_1}$ such that $x_1 = tx$ and $x_1' = tx'$. 

A category has pullbacks if every co-span can be completed to a pullback, which is necessarily unique up to isomorphism.

\begin{defi}
  A \emph{category of worlds}, $\Cscr$, is a category with pullbacks where any span $u : \underline\w \rightarrow \w ,u' :\underline \w \rightarrow \w'$ can be completed to a minimal pullback square. Furthermore, there is a subcategory $\Iscr$ of $\Cscr$ full on objects which is a poset, \ie, $|\Iscr(X,Y)| \leq 1$. The morphisms in $\Iscr$ are called inclusions. Moreover, any morphism $u$ in $\Cscr$ can be factored as $i_1;u_1$ and as $u_2;i_2$ where $i_1,i_2$ are inclusions and $u_1,u_2$ are isomorphisms. 
\end{defi}

\begin{prop}
  In a category of worlds all morphisms are monomorphisms and if $\w\sq{x}{u}{x'}{u'}\w'$ with apex $\overline{\w}$ is a minimal pullback then the morphisms $x$ and $x'$ are \emph{jointly epic}, i.e. for any $f,g : \overline{\w}\to \w_1$, if $fx = gx$ and $fx'=gx'$, then $f=g$. 
\end{prop}
\begin{proof}
First we show that any morphism $u : \w \rightarrow \w'$ is a monomorphism. Let $\w'\sq{x}{u}{x'}{u}\w'$ be a completion of the span $u,u$ to a (minimal) pullback. If $ua = ub =: h$, then $xh = x'h$. So, the pullback property furnishes a unique map $c$ such that $uc = h$. Thus $c = a = b$, so $u$ is a monomorphism.

Now suppose that $\w\sq{x}{u}{x'}{u'}\w'$ is a minimal pullback and $fx = gx =: h$ and $fx' = gx'=: h'$. Then we claim that $\w'\sq{h}{u}{h'}{u'}\w'$ is a pullback: if $ht = h't'$, then since $f,g$ are monomorphisms by the above, we have $xt = x't'$, so we can appeal to the pullback property of the original square.

Minimality of $\w'\sq{x}{u}{x'}{u}\w'$ furnishes a unique map $k$ such that $h = kx$ and $h'= kx'$. But since $f$ and $g$ also have that property ($h = fx$ and $h' = fx'$ and similarly for $g$), we conclude $f = g = k$. 
\end{proof}

\newcommand\Set{\ensuremath{\mathbf{S}}\xspace}

\begin{prop} 
   The {category} $\world$ with 
  finite sets of natural numbers as objects and injective functions for morphisms and inclusions for the subcategory of inclusion ($\Iscr$) is a category of worlds.
\end{prop}
\begin{proof}
  Given $f:X\to Z$ and $g:Y\to Z$ forming a co-span in $\world$, we form
their pullback as $X\xleftarrow{f^{-1}} fX\cap gY \xrightarrow{g^{-1}}
Y$. This is minimal when $fX\cup gY = Z$. Conversely, given a span
$Y\xleftarrow{f} X \xrightarrow{g}Z$, we can complete to a minimal
pullback by \newcommand{\myin}[1]{\mathit{in}_{#1}}
\[
(Y\setminus fX) \uplus fX \xrightarrow{[\myin{1}, \myin{3}\circ f^{-1}]}
(Y\setminus fX) + (Z\setminus gX) + X
\xleftarrow{[\myin{2}, \myin{3}\circ g^{-1}]} (Z\setminus gX) \uplus gX
\]
where $[-,-]$ is case analysis on the disjoint union $Y = (Y\setminus
fX)\uplus fX$.
Thus a minimal pullback square in \world{} is of the form:
\begin{displaymath}
\vcenter{\xymatrix@C=1pc@R=0.5pc{ 
& X_1' \cup X_2'\\
X_1 \cong X_1' \ar[ru]^x & & X_2 \cong X_2'\ar[lu]_{x'}\\
& X_1' \cap X_2' \ar[ru]_{u'} \ar[lu]^u
}}
\end{displaymath}
The factorization property is straightforward.
\end{proof}
An object $\w$ of $\world$ models a set of generated/allocated names, with injective maps corresponding to renamings and extensions with newly generated names.



%

In $\world$, a minimal pullback corresponds to a \emph{partial bijection} between
$X_1$ and $X_2$, as used in other work on logical relations for generativity
\cite{DBLP:conf/mfcs/PittsS93,DBLP:conf/ppdp/BentonKBH07}. We write
$u:x\hookrightarrow y$ to mean that $u$ is a subset inclusion and also use the notation $x\hookrightarrow y$ to denote the subset inclusion map from $x$ to $y$. Of course, the use of this notation implies that $x\subseteq y$. Note
that if we have a span $u,u'$ then we can choose $x,x'$ so that
$\sq{x}{u}{x'}{u'}$ is a minimal pullback and one of $x$ and $x'$ is an inclusion. 
To do that, we simply replace the apex of any minimal pullback completion with an isomorphic one. The analogous property holds for completion of co-spans to pullbacks. 

In this paper, we fix the category of worlds to be $\world$. The general definitions, in particular that of setoid-valued functors that we are going to give, also make sense in other settings. For example, in our treatment of proof-relevant logical relations for reasoning about stateful computation \cite{benton14popl}, we build a category of worlds from partial equivalence relations on heaps.





\section{Setoid-valued functors} 
\label{subsec: func-setoids}

A functor $A$ from the category
of worlds $\world$ to the category of setoids comprises, as usual, for
each $\w\in \world$ a setoid $A\w$, and for each $u:\w\rightarrow \w'$
a morphism of setoids $Au:A\w\rightarrow A\w'$ preserving identities
and composition. This means that there exist continuous functions 
of type $\Pi a. A\w(a, (A \textit{id})\, a)$; and for any 
two morphisms $u : \w \to \w_1$
and $v: \w_1 \to \w_2$ a continuous function of type $\Pi a. A\w_2(Av(Au\, a), A(vu)\, a)$.

If $u:\w\rightarrow \w'$ and $a\in A\w$ we may write $u.a$ or even
$ua$ for $Au(a)$ and likewise for proofs in $A\w$. Note that there is a proof of equality of $(uv).a$ and $u.(v.a)$. In the sequel, we shall abbreviate `setoid-valued
functor(s)' as `\svf(s)'.

Intuitively, \svfs\ will become the denotations of types.
Thus, an element of $A\w$ is a value involving
at most the names in $\w$. If $u:\w\rightarrow \w_1$ then $A\w\ni a\mapsto u.a
\in A\w_1$ represents renaming and possible weakening by names not
``actually'' occurring in $a$. Note that due to the restriction to
injective functions identification of names (``contraction'') is
precluded. This is in line with Stark's use \cite{Stark94thesis} of set-valued functors on
the category $\world$ to model fresh names.

\begin{defi}
\label{def:p.p.f}
We call an \svf,  $A$, \emph{pullback-preserving} if for every
pullback square $\w\sq{x}{u}{x'}{u'}\w'$ with apex $\overline{\w}$ and
low point $\underline{\w}$ the diagram $A\w\sq{Ax}{Au}{Ax'}{Au'}A\w'$
is a pullback in $\Std$. This means that there is a continuous
function of type
\[
\Pi a\in A\w.\Pi a'\in A\w'.A\overline{\w}(x.a,x'.a')\rightarrow \Sigma
\underline{a}\in
A\underline{\w}.A\w(u.\underline{a},a)\times
A\w'(u'.\underline{a},a')
\]
\end{defi}
Thus, if two values $a\in A\w$ and $a'\in A\w'$ are equal in a common world
$\overline{\w}$ then this can only be the case because there is a value in
the ``intersection world'' $\underline{\w}$ from which both $a,a'$ arise. 

\iffull
Note that the ordering on worlds and world morphisms is discrete, so
continuity only involves the $A\w'(u.a,u.a')$ argument. 

The following proposition is proved using a pullback of the form $\sq{u}{v}{u}{v'}$.

\begin{prop}
  If $A$ is a pullback-preserving \svf, $u:\w\rightarrow \w'$ and $a,a'\in A\w$,
  there is a continuous function $A\w'(u.a,u.a')\rightarrow
  A\w(a,a')$. Moreover, the ``common ancestor'' $\underline{a}$ of $a$ and $a'$ is unique up to
$\sim$. 
\end{prop}

All the \svfs\ that we define in this paper will turn
out to be pullback-preserving. However, for the results described in this
paper pullback preservation is not needed. Thus,  we will not use it any
further, but note that there is always the option to require that
property should the need arise subsequently.

Morphisms between functors are  natural transformations in the usual sense; they serve to interpret terms with variables and functions. In more explicit terms, a morphism from \svf\ $A$ to \svf\ $B$ is an equivalence class of pairs
$e=(e_0,e_1)$ where $e_0$ and $e_1$ are continuous functions of the following types:
\[
\begin{array}{l}
 e_0:\Pi\w.A\w\rightarrow B\w \\
  e_1:\Pi\w.\Pi\w'.\Pi
x:\w\rightarrow\w'.\Pi a\in A\w.\Pi a'\in
A\w'.A\w'(x.a,a')\rightarrow B\w'(x.e_0(a),e_0(a'))
\end{array}
\]
Again, the requirements for continuity are simplified by the discrete ordering on worlds.

Two morphisms $e = (e_0,e_1),e' = (e_0',e_1')$ are identified  if there is a continuous function:
\[
  \mu:\Pi \w.\Pi a\in A\w.B\w(e(a),e'(a))
\]
where as in the case of setoids, we omit subscripts where appropriate.
These morphisms compose in the obvious way and so
the \svfs\ and morphisms between them form a category. 
\fi

\section{Instances of  setoid-valued functors}
\label{sec:instances}

We now describe some concrete functors that will
allow us to interpret types of the $\nu$-calculus as \svfs. The simplest one endows any predomain with the structure of an \svf\ where the  equality is proof-irrelevant and coincides with standard equality. The second one generalises the function space of setoids and is used to interpret function types. The third one is used to model dynamic allocation and is the only one that introduces proper proof-relevance. 
\subsection{Base types}
For each predomain $D$ we can define a constant \svf, denoted $D$ as well, with  $D\w$ defined as the discrete setoid over $D$ and $Du$ as the identity. These constant \svfs\ serve as denotations for base types like booleans or integers. 

The \svf\ $N$ of names is given by $N\w = \w$ where $\w$ on the right hand side stands for the discrete setoid over the discrete predomain of names in $\w$, and $Nu = u$ for $u : \w \rightarrow \w'$. Thus, e.g.\ $N\{1,2,3\}= \{1,2,3\}$. 
\subsection{Cartesian closure}
The category of \svfs\ is cartesian closed, which follows from well-known properties of functor categories. The construction of product and function space follows the usual pattern, but we give it here explicitly.

Let $A$ and $B$ be \svfs. The product $A\times B$ is given by taking a
pointwise product of setoids. For the sake of completeness, we note
that $(A\times B)\w=A\w\times B\w$ (product predomain) and $(A\times
B)\w((a,b),(a',b'))=A\w(a,a')\times B\w(b,b')$. This defines a
cartesian product on the category of \svfs. More generally, we can
define the indexed product $\prod_{i\in I}A_i$ of a family $(A_i)_i$ of
\svfs. We write $1$ for the empty indexed product and $()$ for the only element of $1\w$. Note that $1$ is the terminal object in the category of \svfs. 

 The function space $A\Rightarrow B$ is the \svf\ given as follows.
 $|(A\Rightarrow B)\w|$ contains pairs $(f_0,f_1)$ where
 $f_0(u)\in |A\w_1\Rightarrow B\w_1|$ for each $\w_1$ and
 $u:\w\rightarrow\w_1$. If $u:\w\rightarrow \w_1$ and
 $v:\w_1\rightarrow \w_2$ then
\[
f_1(u,v) \in (A\w_1\Rightarrow B\w_2)( [Av\Rightarrow B\w_2]\ f_0(vu), [A\w_1 \Rightarrow Bv]\ f_0(u))
\]
where 
\[\begin{array}{l}
{[}Av\Rightarrow B\w_2] : (A\w_2\Rightarrow B\w_2)  \rightarrow (A\w_1\Rightarrow B\w_2)\\
{[}A\w_1\Rightarrow Bv] : (A\w_1\Rightarrow B\w_1)  \rightarrow (A\w_1\Rightarrow B\w_2)
\end{array}
\]
are the obvious composition morphisms. 

A proof in  $(A\Rightarrow B)\w((f_0,f_1), (f_0',f_1'))$ is a function $g$ that for each $u:\w\rightarrow \w_1$ yields a proof $g(u)\in 
(A\w_1\Rightarrow B\w_1)(f_0(u),f_0'(u))$. 

The order on objects and proofs is pointwise as usual. 
The following is now clear from the definitions. 
\begin{prop}
The category of \svfs\ is cartesian closed.
\end{prop}

\newcommand\Dscr{\mathcal{D}}

We remark that cartesian closure of the category of \svfs\ is an instance of the general results (see~\cite{nlab}) that if $\Dscr$ is cartesian closed and complete, then so is $\Dscr^\Cscr$ for any category $\Cscr$. Here be $\Dscr$ is the category of setoids described in Section~\ref{sec:setoids}.

\begin{defi}
An \svf\ $D$ is \emph{pointed} if $D\w$ is pointed for each $\w$ and the transition maps $Du:D\w\rightarrow D\w_1$ for $u:\w\rightarrow \w_1$ are strict. 
\end{defi}
\begin{thm}\label{fixte}
If $D$ is a pointed \svf\ then there exists a morphism
$Y:(D\Rightarrow D)\rightarrow D$ satisfying the equations from
Theorem~\ref{tfixit} understood relative to the cartesian closed structure of the category of \svfs.
\end{thm}
\begin{proof}
The fixpoint combinator on the level of \svfs\ is defined pointwise. 
Given world $\w$ and $(f_0,f_1)\in (D\Rightarrow D)\w$ we define 
\[
Y\w(f_0,f_1) = Y(f_0(\textit{id}_\w))
\]
where $Y$ is the setoid fixpoint combinator from Theorem~\ref{tfixit}.
The translation of proofs is obvious. We need to show that this
defines a natural transformation. So, let $u:\w\rightarrow \w_1$ and
$(f_0,f_1)\in (D\Rightarrow D)\w$. Put $f:=f_0(\textit{id}_\w)$ and
$g:=f_0(u)$. We need to construct a proof that $Du(Y(f))\sim
Y(g)$. Now, $f_1$ furnishes a proof of $(Du)f=g$ and $Du$ is strict by
assumption on $D$ so that ``Uniformity'' furnishes the desired proof.

The laws from Theorem~\ref{tfixit} can be directly inherited. 
\end{proof}
\begin{defi}
An \svf\ $A$  is \emph{discrete} if $A\w$ is a discrete setoid for every world $\w$. 
\end{defi}
The constructions presented so far only yield discrete \svfs, \ie,
proof relevance is merely propagated, but never actually created. This
is not so for the next operator on \svfs, which is to model dynamic
allocation.

\section{Dynamic Allocation Monad} 

Before we define the dynamic allocation monad we recall Stark's
\cite{Stark94thesis} definition of a dynamic allocation monad for the category of
\emph{set-valued functors} on the category of worlds. For set-valued
functor $A$, Stark defines a set-valued functor $TA$ by
$TA\w=\{(\w_1,a)\mid \w\subseteq \w_1, a\in A\w_1\}/{\sim}$ where
$(\w_1,a)\sim(\w_1',a')$ iff there exist maps $x:\w_1\rightarrow
\overline\w$, $x':\w_1'\rightarrow \overline \w$ for some $\overline
\w$ satisfying $x.i=x'.i'$ and $x.a=x'.a'$ where
$i:\w\hookrightarrow\w_1$ and $i':\w \hookrightarrow \w_1'$ are the
inclusion maps. 

Our dynamic allocation monad for \svfs\ essentially mimics this
definition, the difference being that the maps $i,i'$ witnessing
equivalence of elements now become proofs of
$\sim$-equality. Additionally, our definition is based on predomains and
involves a bottom element for recursion.

\subsection{Definition of the monad}

Let $A$ be an \svf. We put 
\[
|TA\w| = \{(\w_1,a)\mid \w\subseteq \w_1\wedge a\in A\w_1\}_\bot
\]
Thus, a non-bottom element of $TA\w$ consists of an extension of $\w$ together
with an element of $A$ taken at that extension. Note that the
extension is not existentially quantified, but an inherent part of the element.

The ordering is given by $(\w_1,a)\leq (\w_1',a')$ if 
$\w_1=\w_1'$ and $a\leq a'$ in $A\w_1$ and of course, $\bot$ is the
least element of $TA\w$.

The proofs are defined as follows. First,
$TA\w(\bot,\bot)=\{\bot\}$ and second, the elements of
$TA\w((\w_1,a),(\w_1',a'))$ are triples $(x,x',p)$ where $x,x'$
complete the inclusions $u:\w\hookrightarrow \w_1$ and
$u':\w\hookrightarrow \w_1'$ to a commuting square
\begin{displaymath}
\vcenter{\xymatrix@C=1.2pc@R=0.5pc{ 
& \overline{\w} & \\
 \w_1\ar[ur]^{x} & & \w_1'\ar[ul]_{x'} \\
& \w\ar@{^{(}->}[ul]^{u} \ar@{_{(}->}[ur]_{u'}
}}
\end{displaymath}
with $\overline \w=\cod{x}=\cod{x'}$. The third component $p$ then is a
proof that $a$ and $a'$ are equal when transported to $\overline\w$,
formally, $p\in A\overline \w(x.a,x'.a')$. The ordering is again
discrete in $x,x'$ and inherited from $A$ in $p$. Formally,
$((\w_1,a),(\w_1',a'),(x,x',p))$ $\leq$
$((\w_1,b),(\w_1',b'),(x,x',q))$ when $(x.a,x'.a',p)\leq (x.b,x'.b,q)$
in $A \cod x$ and of course $(\bot,\bot,\bot)$ is the least
element. No $\leq$-relation exists between triples with different
mediating co-span. In particular, in an ascending chain of proofs the witnessing spans are always the same, which is the intuitive reason why they can be patched together to form a supremum.

Consider, for example, that $\w=\{0\}$, $\w_1=\{0,1,2\}$,
$\w_1'=\{0,2,3\}$.  Then, both $\cval=(\w_1,\,(0,2))$ and $\cval'=(\w_1',\,(0,3))$ are elements of
$T(N\times N)\w$, and  $(x,x',p)\in T(N\times N)\w(\cval,\cval')$ is a proof that the two are equal
where $x:\w_1\rightarrow\overline \w=\{0,1,2,3\}$ sends $0\mapsto 0$,
$1\mapsto 1$, $2\mapsto 2$ and $x':\w_1'\rightarrow\overline \w$ sends
$0\mapsto 0$, $2\mapsto 3$, $3\mapsto 2$. The proof $p$ is the
canonical proof by reflexivity. Note that, in this case, the order relation is
trivial. It becomes more interesting when the type of values $A$ is a
function space. 

Next, we define the morphism part. 
Assume that $u : \w \rightarrow \q$ is a morphism in $\world$. We want to construct a
morphism $Au : TA \w \to TA\q$ in $\Std$. So let $(\w_1,a)\in TA\w$ and $i:\w\hookrightarrow \w_1$ be the inclusion. We complete the span $i,u$ to a minimal pullback 
\begin{displaymath}
\vcenter{\xymatrix{ 
\w_1\ar[r]^{u_1} & \q_1\\
\w\ar[r]_u\ar@{^{(}->}[u]^i& \q\ar@{^{(}->}[u]_{j}}}
\end{displaymath}
with $j$ an inclusion as indicated. We then define
$TAu(\w_1,a)=(\q_1,u_1.a)$. 
We
assume a function that returns such completions to minimal pullbacks
in some chosen way. The particular choice is unimportant.

Picking up the previous example and letting $u:\w\rightarrow \q=\{0,1\}$ be $0\mapsto 1$ then a possible completion to a minimal pullback would be 
\begin{displaymath}
\vcenter{\xymatrixcolsep{5pc}\xymatrix{ 
\w_1=\{0,1,2\}\ar[r]^{0\mapsto 1,1\mapsto 2,2\mapsto 3} & \{0,1,2,3\}=\q_1\\
\w=\{0\}\ar[r]_{0\mapsto 1}\ar@{^{(}->}[u]^i& \{0,1\}=\q\ar@{^{(}->}[u]_{j}}}
\end{displaymath}
Note that the following square where the additional name $1$ in $\q$ is identified with a name already existing in $\w_1$ is \emph{not} a pullback 
\begin{equation}
\vcenter{\xymatrixcolsep{5pc}\xymatrix{ 
\w_1=\{0,1,2\}\ar[r]^{0\mapsto 1,1\mapsto 0,2\mapsto 3} & \{0,1,2,3\}\\
\w=\{0\}\ar[r]_{0\mapsto 1}\ar@{^{(}->}[u]^i& \{0,1\}=\q\ar@{^{(}->}[u]}}
\label{eq:comm-squa}
\end{equation}
Adding extra garbage into $\q_1$ like so would result in a pullback that is not minimal. 
\begin{displaymath}
\vcenter{\xymatrixcolsep{5pc}\xymatrix{ 
\w_1=\{0,1,2\}\ar[r]^{0\mapsto 1,1\mapsto 2,2\mapsto 3} & \{0,1,2,3,4,5\}\\
\w=\{0\}\ar[r]_{0\mapsto 1}\ar@{^{(}->}[u]^i& \{0,1\}=\q\ar@{^{(}->}[u]}}
\end{displaymath}
If $(x,x',p)$ is a proof of $(\w_1,a)\sim(\w_1',a')$
then we obtain a proof, $(\q_1',u_1'.a')$, that $TAu(\w_1,a)\sim TAu(\w_1',
a')$ as follows. We first complete
the span $xi,u$ to a minimal pullback with apex $\overline \q$ and upper arrow
$\overline u:\cod x=\overline \w\rightarrow \overline \q$. Now minimality of the pullbacks apexed at $\q_1$ and $\q_1'$ furnishes
morphisms $y:\q_1\rightarrow \overline \q$ and $y':\q_1'\rightarrow \overline \q$ so that $yj=y'j'$ (where 
$j':\q\hookrightarrow \q_1'$). We then have $(y,y',\overline u.p):TAu(\w_1,a)\sim TAu(\w_1',a')$ as required. This
shows that the passage $(\w_1,a)\mapsto (\w_1',a')$ actually does define a
morphism of setoids.
\begin{displaymath}
\vcenter{\xymatrix@C=1pc@R=0.5pc{ 
& \overline{\w} \ar[rrr]^{\overline{u}} & & & \overline{\q} & \\
& & \w_1' \ar[ul]_{x'} \ar[rrr]^<(0.15){u_1'} & & & \q_1' \ar[ul]_{y'}\\
\w_1 \ar[uur]^x \ar[rrr]_<(0.6){u_1}|<(0.31)\hole & & & \q_1 \ar[uur]^y|\hole \\ \\
\w\ar@{_{(}->}[uuurr]_{i'} \ar@{^{(}->}[uu]^i \ar[rrr]_{u}& & & \q \ar@{_{(}->}[uuurr]_{j'} \ar@{^{(}->}[uu]^j
}}
\end{displaymath}
The functor laws amount to similar constructions of $\sim$-witnesses
and are left to the reader.
The following is direct from the definitions. 
\begin{prop}
$T$ is a strong monad on the category of \svfs. The unit sends $\val\in A\w$
to $(\w,\val)\in (TA)\w$. The multiplication sends 
$(\w_1,(\w_2,v))\in (TTA)\w$ to $(\w_2,v)\in TA\w$. The strength map sends $(a,(\w_1,b))\in (A\times TB)\w$ to $(\w_1,(a.i,b))$ where $i:\w\hookrightarrow\w_1$. 

Notice that if we had taken any arbitrary commuting square, like the one shown in Equation~\ref{eq:comm-squa}, then preservation of proofs could not be guaranteed because names in extensions would be captured in an arbitrary way. Requiring minimality, on the other hand, is merely a technical convenience.

\end{prop}
\subsection{Comparison with cpo-valued functors}
The flawed attempt at defining a dynamic allocation monad for
FM-domains discussed by Shinwell~\cite{shinwell04phd} 
and mentioned in the introduction can
be reformulated in terms of cpo-valued functors and further highlights
the importance of proof-relevant equality. 

\newcommand\tsp{T_{\! \! sp}} 

Given a cpo-valued functor $A$ one may construct a poset-valued
functor $\tsp A$ which has for underlying set equivalence classes of
pairs $(\w_1,a)$ with $\w \subseteq \w_1$ and $a \in A\w_1$. As in Stark's definition above, we have a $(\w_1,a) \sim (\w_1',a')$ if there are morphisms $x,x'$ such that $xi = x'i'$ and $x.a = x'.a'$ where $i : \w \rightarrow \w_1$, $i':\w \rightarrow \w_1'$ are the inclusions. As for the ordering,
the only reasonable choice is to decree that on representatives
$(\w_1,a)\leq (\w_1',a')$ if $x.a\leq x'.a'$ for some co-span $x,x'$
with $xi=x'i'$ where $i,i'$ are the inclusions as above.  However,
while this defines a partial order it is not clear why it should have
suprema of ascending chains because the
witnessing spans might not match up so that they can be pasted to a witnessing span for the limit of the chain. Indeed, Shinwell's
thesis~\cite{shinwell04phd} contains a concrete counterexample, which is due to Pitts.

In our notation, Pitts's counterexample takes the following form. Define the cpo-valued functor $A$ by $A\w := (\Pscr(\w),\subseteq)$. So the elements of $A\w$ are subsets of $\w$ ordered by inclusion, hence a finite cpo. Let us now examine $\tsp A$. An element of $\tsp A\w$ is an $\sim$-equivalence class of pairs $(\w_1,U)$ where $U \subseteq \w_1$, $\w \subseteq \w_1$. Furthermore, $(\w_1,U) \sim (\w_1',U')$ whenever $U = U'$ and the ordering $\leq$ on $\tsp A\w$ is $(\w_1,U) \leq (\w_1',U')$ whenever $U \subseteq U'$. Let $t_n$ be the equivalence class of $(\{0,\ldots,n-1\},\{0,\ldots,n-1\})$. We have $t_n \in \tsp A\emptyset$ for all $n$ and $t_n \leq t_m\iff n \leq m$. From this it is clear that the ascending chain $t_0 \leq t_1 \leq \cdots$ does not have a least upper bound in $\tsp A\emptyset$ for if $(\w_1,U)$ were such an upper bound then $|U| \geq n$ would have to hold for all $n$.

The transition to proof relevance that we have made allows us to define the order on representatives as we have done and thus to bypass these difficulties. We view $A$ above as an \svf\ with underlying cpo $A\w = \w$ and, trivial, \ie, discrete equality. Now applying our dynamic allocation monad $T$ to $A$ yields the \svf\ $TA\w$ whose underlying cpo contains in addition to $\bot$, pairs $(\w_1,U)$ where $U \subseteq \w_1$ with ordering $(\w_1,U) \leq (\w_1',U')$ if $\w_1 = \w_1'$ and $U \subseteq U'$. A proof that an element $(\w_1,U)$ is equal to the element $(\w_1',U')$ is given by a triple $(\w_2,u,u')$ such that $u :\w_1 \rightarrow \w_2$ and $u': \w_1' \rightarrow \w_2$ and moreover $u(U) = u'(U')$. The ordering in these proofs is the discrete one. Now the sequence shown above is not an ascending chain and thus is no longer a counter-example to completeness.


\section{Observational Equivalence and Fundamental Lemma}
\label{sec:prlr}
We now construct the machinery that connects the concrete language with the 
denotational machinery introduced in Section~\ref{sec:semantics}.
The semantics of types, written using $\sem{\cdot}{}$, as \svfs\ is defined inductively as follows:
\begin{itemize}
 \item For basic types $\sem{\tau}{}$ is the corresponding discrete \svf.
 \item $ \sem{\tau \to \tau'}$ is defined as the function space $\sem{\tau} \to T \sem{\tau'}$, 
 where $T$ is the dynamic allocation monad. 
\item For typing context $\Gamma$ we define $\sem{\Gamma}$ as the indexed product of \svfs\ $\prod_{x\in\dom{\Gamma}}\sem{\Gamma(x)}$. 
\end{itemize}

To each term in context $\Gamma\vdash e:{\tau}$ we can associate a
morphism $\sem{e}$ from $\sem{\Gamma}$ to $T\sem{\tau}$ by
interpreting the syntax in the category of \svfs\ using cartesian
closure, the fixpoint combinator, and the fact that $T$ is a strong monad. We
omit most of the straightforward but perhaps slightly tedious definition and
only give the clauses for ``new'' and ``let'' here:
\[
\sem{\New}\w = (\w\cup\{n+1\},n+1)
\]
where $n = \max(\w)$ and 
$\max(\w)=\max(\{n\mid n\in\w\})$, \ie, the greatest number in the world $\w$.

If $f_1:\sem{\Gamma}\rightarrow T \sem{\tau_1}$ and $f_2:\sem{\Gamma,x{:}\tau_1}\rightarrow T\sem{\tau_2}$ are the denotations of $\Gamma\vdash e_1:\tau_1$ and 
$\Gamma,x{:}\tau_2\vdash e_2:\tau_2$ then the interpretation of $\letin{x}{e_1}{e_2}$ is the morphism $f:\sem{\Gamma}\rightarrow T\sem{\tau_2}$ given by 
\[
f  =  \mu \circ T f_2 \circ \sigma \circ \langle \textit{id}_\Gamma,f_1\rangle
\]
where $\mu$ is the monad multiplication, $\sigma$ is the monad strength and where we have made the simplifying assumption that $\sem{\Gamma,x{:}\tau}=\sem{\Gamma}\times\sem{\tau}$. Assuming that $f_1$ and $f_2$ now stand for the first components of concrete representatives of these morphisms, one particular concrete representative of this morphism (now also denoted $f$) satisfies: 
\[\begin{array}{l}
f\w(\gamma) = f_2(i.\gamma,a), 
\mbox{where }i \mbox{ is the inclusion } \w{\hookrightarrow} \w_1 \mbox { and } f_1\w(\gamma)=(\w_1,a).\\
\end{array}\]
Our aim is now to relate these morphisms to the computational interpretation $\semV{e}$. 
\begin{defi}
\label{def:realizability}
For each type $\tau$ and world $\w$ we define two relations; the relation $\Vdash^\tau_\w \subseteq \semV{\tau}\times \sem{\tau}\w$ and $\Vdash^{T\tau}_\w \subseteq (\mathbb{N} \rightarrow (\mathbb{N} \times \semV{\tau})_\bot) \times T\sem{\tau}\w$ by the following clauses.  
\[\begin{array}{l}
\bbval \Vdash^{\booltype}_\w \bval \iff \bbval = \bval\\
\iival \Vdash^{\inttype}_\w \ival \iff \iival = \ival\\
l \Vdash^{\loctype}_\w k \iff l=k\\
f\Vdash^{\tau\to \tau'}_\w g\iff \forall \w_1\supseteq\w.\forall \vval. \forall \val.
\vval\Vdash^\tau_{\w_1} \val \Rightarrow f(\vval)\Vdash_{\w_1}^{T \tau'} g_0(\w{\hookrightarrow}\w_1,\val)\\[5pt]
\ccval\Vdash^{T\tau}_\w \cval\iff 
\begin{array}{l}
    [\ccval(\max(\w)+1)=\bot \Leftrightarrow \cval=\bot]~ \land\\
 {[}\ccval(\max(\w)+1)= (n_1,\vval) \land 
\cval = (\w_1, \val) \Rightarrow 
       n_1 =\max(\w_1)+1 \wedge \vval \Vdash_{\w_1}^\tau \val))].
\end{array}
\end{array}
\]
\end{defi}
Notice that $T\tau$ is not part of the syntax, but $T$ is a marker to distinguish the two relations defined above.

The following lemma states that 
the realizability relation is stable with respect to enlargement of worlds. It is needed for the ``fundamental lemma'' \ref{fundi}.
\begin{lem}
\label{lemma:contexts}
Let $\tau$ be a type. 
If $u:\w\hookrightarrow \w_1$ is an inclusion as indicated and $\vval\Vdash^\tau_\w \val$ then $\vval \Vdash^\tau_{\w_1}u.\val$, too. 
\end{lem}
The proof is by a straightforward induction on types. Note, however, that the restriction to inclusions is important for the cases of function type and the type $\loctype$.
We extend $\Vdash$ to typing contexts by putting 
\[
\eta\Vdash^\Gamma_\w\gamma\iff \forall x\in\dom{\Gamma}.\eta(x)\Vdash^{\Gamma(x)}_\w \gamma(x)
\]
for $\eta\in\semV{\Gamma}$ and $\gamma\in\sem{\Gamma}$. 
\begin{thm}[Fundamental lemma]\label{fundi}
Let $\Gamma\vdash e:\tau$ be a well typed term. There exists a representative 
$(\cval,\_)$ of the  equivalence class $\sem{e}$ at world $\w$ such that 
if  $\eta\Vdash^\Gamma_\w\gamma$ then 
$\semV{e}\eta\Vdash^{T\tau}_\w \cval (\gamma)$. 
\end{thm}
\begin{proof}
By induction on typing rules. We always chose for the representative the one given as witness in the definition of $\sem{e}$.
Most of the cases are straightforward. For illustration we show $\New$ and  $\lett$:
As for $\New$, we pick the representative $\cval$ that at world $\w$ returns $(\w\cup\{\max(\w) +1\},\max(\w))$. Now, with 
$\ccval=\semV{\New}$, we have $\ccval(\max(\w))=(\max (\w)+1,\max(\w))$ and 
$\ccval\Vdash_\w^{TN}\cval$ holds, since $\max(\w\cup\{\max(\w) +1\})=\max(\w) +1$.

Next, assume that $\Gamma\vdash \letin{x}{e_1}{~e_2} :\tau_2$, where $\Gamma\vdash e_1:\tau_1$ and 
$\Gamma, x: \tau_1 \vdash e_2 : \tau_2$. Choose, according to the induction hypothesis appropriate 
representatives $\cval_1$ of $\sem{e_1}$ and $\cval_2$ of $\sem{e_2}$. If  
$\eta\Vdash^\Gamma_\w\gamma$ for some initial world $\w$ then we have (H1) $\semV{e_1}\eta\Vdash^{T\tau_1}_\w \cval_1(\gamma)$. 
 If $\semV{e_1}\eta(\max(\w) + 1)=\bot$ then $\cval_1(\gamma)=\bot$, too, and the same goes for the interpretation of the 
 entire let-construct. So suppose that $\semV{e_1}\eta(\max(\w)+1)=(n_1,\vval)$. By (H1), we must then 
 have $\cval_1(\gamma)=(\w_1,\val)$ where 
$\w\subseteq\w_1$ and $n_1=\max(\w_1)+1$ and $\vval\Vdash^{\tau_1}_{\w_1}\val$. 

By Lemma~\ref{lemma:contexts} we then have $\eta\Vdash^\Gamma_{\w_1} i.\gamma$ where $i:\w\hookrightarrow \w_1$. Thus, by the induction hypothesis, we get
 (H2) $\semV{e_2}(\eta[x{\mapsto}\vval]) \Vdash^{T\tau_2}_{\w_1} \cval_2(i.\gamma[x{\mapsto}\val])$. Thus, putting $\cval(\gamma)=\cval_2(i.\gamma,\val)$ furnishes the required representative of $\sem{\letin{x}{e_1}{e_2}}(\w)$.
\end{proof}
\begin{rem}
Note that the particular choice of representative matters here. For example, if 
$\cval_0\w = (\w\uplus \{\max(\w)+1,\max(\w)+2\},\max(\w)+1)$ then there exists $\cval_1$ such that $(\cval_0,\cval_1):1\rightarrow TN$ and $(\cval_0,\cval_1)$ and $\sem{\New}$ are equal qua morphisms of \svfs. Yet, $\semV{\New}\not\Vdash_\w^{TN}\cval_0$. 

It would have been an option to refrain from the identification of $\sim$-related morphisms. The formulation of the Fundamental Lemma would then have become slightly easier as we would have defined $\sem{e}$ so as to yield the required witnesses directly. On the other hand, the equational properties of the so
 obtained category  would be quite weak and in particular cartesian closure, monad laws, functor laws, etc would only hold up to $\sim$. This again would not really be a problem but prevent the use of standard category-theoretic terminology. 
\end{rem}
\subsection{Observational Equivalence}

\begin{defi}
  Let $\tau$ be a type.  We define an
  \emph{observation of type $\tau$} as a closed term  $\vdash o:\tau \rightarrow \booltype$. 
  Two values $\vval,\vval'\in\semV{\tau}$ are \emph{observationally equivalent at type $\tau$}
  if for all observations $o$ of type $\tau$ one has that $\semV{o}(\vval)(0)$ is
  defined iff $\semV{o}(\vval')(0)$ is defined and when
  $\semV{o}(v)(0)=(n_1,\vval_1)$ and $\semV{o}(v')(0)=(n_1',\vval_1')$
  then $\vval_1=\vval_1'$.
\end{defi}
Note that observational equivalence is a congruence since an observation can be extended by any englobing context. We also note that observational equivalence is the coarsest reasonable congruence. 

We now show how the proof-relevant semantics can be used to deduce observational equivalences. 
\begin{thm}[Observational equivalence]
\label{thm:obs-equivalence}\label{obseq}
If $\tau$ is a  type and $\vval\Vdash^\tau_{\emptyset} e$ and
$\vval'\Vdash^\tau_{\emptyset}e'$ with $e\sim e'$ in $\sem{\tau}\emptyset$ then $\vval$ and  $\vval'$ are
observationally equivalent at type $\tau$. 
\end{thm}
\begin{proof}
Let $o$ be an observation at type $\tau$. By the Fundamental Lemma (Theorem \ref{fundi}) we have  $\semV{o}\Vdash_\emptyset^{\tau\to \booltype}\sem{o}$. 

Now, since $e\sim e'$ we also have $\sem{o}(e)\sim\sem{o}(e')$ and, of course, 
$\semV{o}(\vval)\Vdash_\emptyset^{T\booltype}\sem{o}(e)$ and 
$\semV{o}(\vval')\Vdash_\emptyset^{T\booltype}\sem{o}(e')$. 

Knowing $\sem{o}(e)\sim\sem{o}(e')$\footnote{More precisely, we are using the representative of the equivalence class given by Theorem~\ref{fundi}.}, there are two cases. If $\sem{o}(e)(0)$ and  $\sem{o}(e')(0)$ both diverge, then the same is true for 
$\semV{o}(\vval)(0)$ and $\semV{o}(\vval')(0)$ by definition of $\Vdash^{T\booltype}$. Alternatively, if $\sem{o}(e)(0)=(\_,\_,b,\_)$ and $\sem{o}(e')(0)=(\_,\_,b',\_)$ for booleans $b,b'$ then, by definition of $\sim$ at $T\sem{\booltype}$ we get $b=b'$ and, again by definition of $\Vdash^{T\booltype}$, this then implies that 
$\semV{o}(\vval)(0)=(\_,b)$ and $\semV{o}(\vval)(0)=(\_,b')$ with $b=b'$, hence the claim. 
\end{proof}
\section{Direct-Style Proofs}\label{proffs}
\label{sec:examples}
We now have enough machinery to provide  direct-style proofs for equivalences involving
name generation.

If $\Gamma \vdash e : \tau$ and  $\Gamma \vdash e' : \tau$, we say the equation $\Gamma \vdash e = e' :\tau$ is \emph{semantically sound} if $\sem{e} = \sem{e'}$ are equal morphisms from $\sem{\Gamma}$ to $\sem{\tau}$. If $v = v'$ can be derived by sound equations and congruence rules, then $\sem{v}$ and $\sem{v'}$ are equivalent by Theorem~\ref{obseq}. We omit the formal definition of such derivations using an equational theory. We refer to \cite{benton14popl} for details on how this could be set up.  

From the categorical properties of setoids the soundness of $\beta, \eta$, fixpoint unrolling and similar equations is obvious. We now demonstrate the soundness of the more interesting equations involving name generation.


\subsection{Drop equation}
We start with the following equation, which eliminates a dummy allocation:
\[
   (\letin{x}{\New} e) = e,\quad \textrm{ provided $x$ is not free in $e$.}
\]
Formally we have $\Gamma \vdash e : \tau$ and, writing $\ccval$ for the LHS of the above equation, and $\ccval'$ for the RHS, the equation reads $\Gamma \vdash \ccval = \ccval' : \tau$. We have $\sem{\ccval'}\w(\gamma) = (\w_1,v)$ for some extension $\w_1$ of $\w$ and $v: \sem{\tau}\w_1$ and $\sem{c}\w(\gamma) = (\w_2,i.v)$ where $\w_2 = \w_1 \cup \{\max(\w_1) + 1\}$
and $i : \w_1 \hookrightarrow \w_2$. 

Now it remains to construct a proof of $(\w_1,v) \sim (\w_2,v) \in TA\w$, which should depend continuously on $\gamma$. To that end, we consider the following pullback square, where
the annotations above and below the square are just to illustrate in which 
world the semantic values are:
\begin{displaymath}
\vcenter{\xymatrix@C=1pc@R=0.6pc{
\sem{c'}\gamma & \val\\
& \w_1 \ar@{^{(}->}[dr]^i \\
\w \ar[ur] \ar[dr] & & \w_2\\
& \w_2\ar[ur]_{id}\\
\sem{c} \gamma & i.\val
}}
\end{displaymath}
Clearly we have $i.\val \sim id.i.\val$ and therefore the pullback above is a proof 
that $(\w_1,v) \sim (\w_2,v) \in TA\w$.

\subsection{Swap equation}
Let us now consider the following equivalence where the order in which the names are generated is 
switched:
\[
(\letin{x}{\New} \letin{y}{\New} e) = (\letin{y}{\New} \letin{x}{\New} e).
\]
Again, we write $\ccval$ for the LHS and $\ccval'$ for the RHS. Let $\cloc_1,\cloc_2,\cloc_1',\cloc_2'$ be the concrete locations allocated by the left-hand-side and right-hand-side of the equation. In fact, $\cloc_1 = \max(\w) + 1, \cloc_2 = \cloc_1 + 1$ and $\cloc_2' = \cloc_1$ and $\cloc_1' = \cloc_2$. We have $\sem{c}\gamma = (\w_2,v)$, where $\sem{e}(i.\gamma[x \mapsto \cloc_1, y \mapsto \cloc_2 ] ) = (\w_2,v)$. We also have $\sem{c'}\gamma$, where $\sem{e}(i'.\gamma[x \mapsto \cloc_2', y \mapsto \cloc_1' ] ) = (\w_2',v')$. 

Define $s(\cloc_1) = \cloc_2'$, $s(\cloc_2) = \cloc_1'$ and $s \upharpoonleft \w = id$. Naturality of $\sem{e}$,  \ie, $\sem{e} \circ \sem{\Gamma,x,y}s \sim T\sem{t}s\circ\sem{e}$ furnishes a co-span $x,x'$ so that $x.s_2.v \sim x'.v'$ and $xt = x'u_1'$ (III). Here $s_2,t$ is the completion of the span $s_1,u_1$ to a minimal pullback as contained in the definition of $T\sem{\tau}s$.
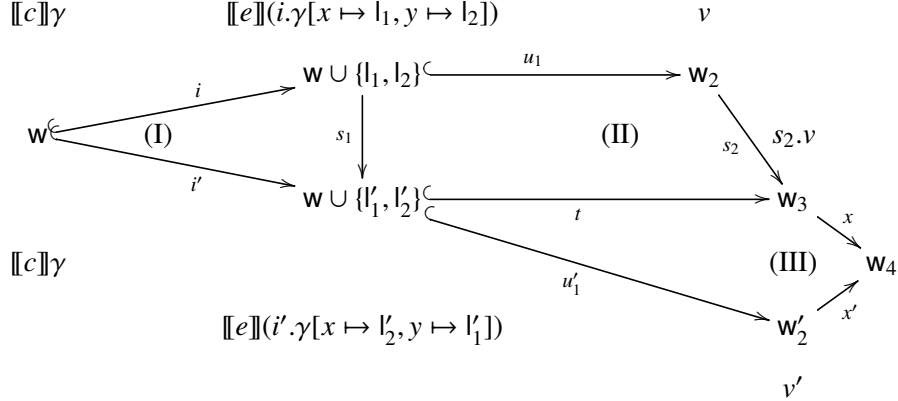
\begin{displaymath}
\vcenter{\xymatrix@C=1pc@R=0.6pc{
\sem{c}\gamma & & \sem{e}(i.\gamma[x \mapsto \cloc_1, y \mapsto \cloc_2 ] ) &&& v \\
&& \w \cup \{\cloc_1,\cloc_2\} \ar[dd]_{s_1} \ar@{^{(}->}[rrr]^{u_1} & &  & \w_2 \ar[ddr]_{s_2} &    \\
\w \ar@{^{(}->}[urr]^i\ar@{^{(}->}[drr]_{i'} & \quad \textrm{(I)} &  &  & \textrm{(II)} & &s_2.v    \\
&& \w \cup \{\cloc_1',\cloc_2'\} \ar@{^{(}->}[ddrrrr]_{u_1'} \ar@{^{(}->}[rrrr]_{t} & & & & \w_3 \ar[dr]^x\\
\sem{c}\gamma &&&&&&\textrm{(III)} & \w_4\\
 && \sem{e}(i'.\gamma[x \mapsto \cloc_2', y \mapsto \cloc_1' ] ) &&&& \w_2' \ar[ur]_{x'}\\
 &&&&&& v'  
}}
\end{displaymath}
Notice that the square (I) commutes by definition of $s_1$; the square (II) commutes because it is a minimal pullback. As a results the entire diagram commutes. $xs_2u_1$ and $x'u_1'$ is the proof that $\sem{c}\gamma \sim \sem{c'}\gamma$.

This is essentially the same proof as given by Stark \cite{Stark94thesis}, but now it also works in the presence of recursion.

\section{Proof-relevant parametric functors}
\label{sec:parametric}
The following equation (Stark's ``Equivalence 12'') 
cannot be validated in the functor category model and nor is it valid in the category of \svfs. 
\begin{equation}
(\letin{n}{\New} \vfun{x}{x=n}) = (\vfun{x}{\mfalse}).  
\label{eq:equiv-12}
\end{equation}
The above is, nevertheless, a valid contextual equivalence. The
intuition is that the name $n$ generated in the left-hand side is
never revealed to the context and is therefore distinct from any name
that the context might pass in as argument to the function; hence, the
function will always return $\mfalse$. To justify this equivalence,
Stark constructs a model based on traditional Kripke logical
relations. He also gives a category-theoretic version of that logical
relation using so-called parametric functors. In this section, we construct a
proof-relevant version of these parametric functors, which will allow
us to justify the above equivalence in the presence of recursion and
in direct style. In fact, this seems to be the first time that this
equivalence has been established in this setting; we are
not aware of an earlier extension of parametric functors to recursion.

We also show that the transition to proof relevance makes the induced
logical relation transitive, which is generally not the case for
ordinary Kripke logical relations.


\subsection{Spans of Worlds}

We use capital letters $S,S', \dots$ for spans of worlds.
If $S$ is the span $\spn{\w}{u}{\underline\w}{u'}{\w'}$ then we use the notations $S:\w\leftrightarrow \w'$ and $\w=\domL S$ (left domain), $\w'=\domR S$ (right domain), $\underline \w=\lop S$ (low point), $u=S.u$, $u'=S.u'$. For world $\w$ we denote $r(\w):\w\leftrightarrow\w$
 the identity span $\spn{\w}{\id}{\w}{\id}{\w}$. If $S:\w\leftrightarrow \w'$ then $s(S):\w'\leftrightarrow\w$ is given by $\spn{\w'}{S.u'}{\lop S}{S.u}{\w}$. If $S:\w\leftrightarrow \w'$ and $S':\w'\leftrightarrow \w''$ then we define $t(S,S'):\w\leftrightarrow \w''$ as $\lspn{\w}{S.u\ x}{.}{S'.u'\ x'}{\w''}$ where $x, x'$ complete $S.u'$ and $S'.u$ to a pullback square. 
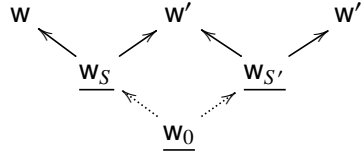
\begin{displaymath}
\vcenter{\xymatrix@C=1pc@R=0.5pc{ 
& \w & & \w'& &  \w'' \\
& & \underline{\w_S} \ar[ul]\ar[ur] & & \underline{\w_{S'}} \ar[ur] \ar[ul] \\
& && \underline{\w_0} \ar@{.>}[ul]\ar@{.>}[ur]
}}
\end{displaymath}
We assume a fixed choice of such completions to pullback squares. We do not assume that the $t$-operation is associative or satisfies any other laws. 

\begin{defi}\label{parames}
A \emph{parametric square} consists of two spans $S:\w\leftrightarrow \w'$ and $S_1:\w_1\leftrightarrow \w_1'$ and two morphisms $u:\w\rightarrow \w_1$ and $u':\w'\rightarrow\w_1'$ such that there exists a morphism $\underline u$ making the two squares in the following diagram pullbacks (thus in particular commute). 
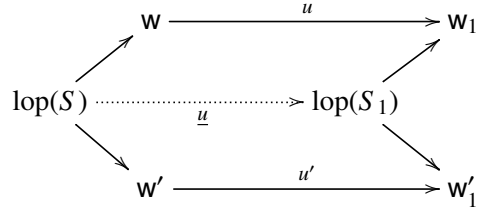
\begin{displaymath}
\vcenter{\xymatrix@C=1pc@R=1.3pc{ 
& \w \ar[rrrr]^{u} & &  & & \w_1  \\
{\lop S} \ar@{.>}[rrrr]_{\underline{u}} \ar[ur]\ar[dr] & &   & & {\lop{S_1}} \ar[ur] \ar[dr] \\
& \w' \ar[rrrr]^{u'} & & & & \w_1'
}}
\end{displaymath}

We use the notation $(u,u'):S\rightarrow S_1$ in this situation. 
\end{defi}
Note that the witnessing morphism $\underline u$ is uniquely determined since  we can complete $S_1$ to a pullback in which case $\underline u$ is the unique mediating morphism given by universal property of the latter pullback.

The reader is invited to check that under the interpretation of spans as partial bijections the presence of a parametric square $(u,u'):S\rightarrow S'$ asserts that $S'$ is obtained from $S$ by consistent renaming followed by the addition of  links and ``garbage''. In the following diagram the left diagram is parametric and the right one is not. 
In particular, the value 2 is mapped to 2 and 3 in the diagram to the right (as illustrated in red).
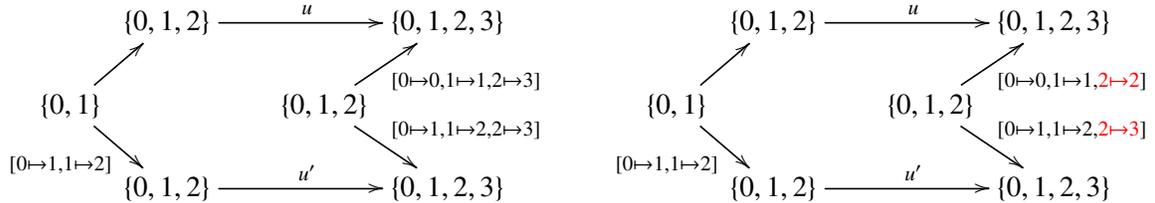
\begin{displaymath}
\vcenter{\xymatrix@C=0.2pc@R=1.3pc{ 
& \{0,1,2\} \ar[rrrr]^{u} & &  & & \{0,1,2,3\}  \\
\{0,1\} \ar[ur]\ar[dr]_{[0\mapsto 1,1\mapsto 2]~~ } & &   & & \{0,1,2\} \ar[ur]_{\quad [0\mapsto 0,1\mapsto 1,{2\mapsto 3}]} 
\ar[dr]^{\quad [0\mapsto 1,1\mapsto 2,2\mapsto 3]} \\
& \{0,1,2\} \ar[rrrr]^{u'} & & & & \{0,1,2,3\}
}}
\qquad
\vcenter{\xymatrix@C=0.2pc@R=1.3pc{ 
& \{0,1,2\} \ar[rrrr]^{u} & &  & & \{0,1,2,3\}  \\
\{0,1\} \ar[ur]\ar[dr]_{[0\mapsto 1,1\mapsto 2]~~ } & &   & & \{0,1,2\} \ar[ur]_{\quad [0\mapsto 0,1\mapsto 1,{\color{red} 2\mapsto 2}]} 
\ar[dr]^{\quad [0\mapsto 1,1\mapsto 2,{\color{red} 2\mapsto 3}]} \\
& \{0,1,2\} \ar[rrrr]^{u'} & & & & \{0,1,2,3\}
}}
\end{displaymath}
 We also note that if $S,S':\w\leftrightarrow \w'$ then
 $(\id,\id):S\rightarrow S'$ is a parametric square if and only if
 there exists an isomorphism $t:\lop S\rightarrow \lop{S'}$ such that
 $S'.u\ t=S.u$ and $S'.u'\ t = S.u'$.  In this case, we call $S$ and
 $S'$ isomorphic spans and write $S \cong S'$. Notice that for any
 span $S:\w\leftrightarrow\w'$ we have $t(S,r(\w')) \cong S \cong
 t(r(\w),\w')$ as well as other properties, such as associativity of
 $t(\cdot,\cdot)$ up to $\cong$.

\begin{defi}
\label{def:stark-parametric}
   A \emph{parametric functor} is a set-valued functor on the category
   of worlds (a set $A\w$ for each world $\w$ and functorial
   transition functions $Au: A\w\rightarrow A\w'$ when
   $u:\w\rightarrow \w'$) together with a relation $AS\subseteq
   A\w\times A\w'$ for each span $S:\w\leftrightarrow \w'$. It is
   required that $Ar(w)$ is the equality relation on $A\w$.
 \end{defi} 

\subsection{Parametric Setoid-Valued Functors}

Our aim is now to define a proof-relevant version of parametric
functors: parametric \svfs. One way to go about this would be to
identify a relational structure, i.e.~an internal reflexive graph in
the category of setoids, and then define parametric functors in the
way that outlined by Stark \cite[Section~4.3]{Stark94thesis}.
Here, we prefer to take
a different approach which gives a slightly richer
structure. Namely, we consider relations that admit composition and
taking opposites. This will lead to a logical relation that is
both transitive and symmetric, and it appears that
proof-relevance plays an important role in making that possible.

This construction could be given ``decent categorical
credentials'' \cite[Section 4.3]{Stark94thesis} by identifying
an internal bicategory (as opposed to a reflexive graph) in the
category of setoids, but we prefer to give a direct and elementary
presentation which has the additional advantage that we can economise
some data, namely the proof components of the relevant setoids, as
those can be recovered from the relational structure by taking the
identity extension property ($Ar(\w)=\textit{id}(A\w)$) as the very
definition of (proof-relevant!) equality on the cpo $A\w$.

So, just like an \svf, a parametric \svf, $A$, specifies a predomain
$A\w$ for each $\w$, and for each $u: \w \rightarrow \w'$ a continuous
function $Au : A\w \rightarrow A\w'$. This time, however, we have a
``heterogeneous equality'' allowing one to compare elements of two
different worlds without the need to transport them to a larger
common world as done in \svfs. Thus, a parametric \svf\ has for each
span $S: \w \leftrightarrow \w'$ and elements $a \in A\w, a'\in A\w'$,
a set of ``proofs'' $AS(a,a')$ asserting equality of these
elements. As in the case of \svfs, the set of tuples $(S,a,a',p)$
with $p \in AS(a,a')$ must carry a predomain structure. We also
require this semantic equality to be reflexive, symmetric and
transitive in an heterogeneous sense, thus employing the $r,s,t$
operations on spans defined above.  Furthermore, the transition
functions should behave functorially, as for \svfs, but this time in the
sense of the ``heterogeneous equality''. Every \svf\ gives rise to
a parametric \svf\ by instantiating the heterogeneous equality to the
larger world, but not all parametric \svfs\ are of this form (see
Example~\ref{myes}).

\begin{defi}\label{defpa}
A \emph{parametric \svf}, $A$, consists of the following data. 
\begin{enumerate}
\item For each world $\w$ a predomain $A\w$. 
\item For each $u:\w\rightarrow\w'$ a continuous function $Au: A\w\rightarrow A\w'$. We use the notation $u.a=Au(a)$. 
\item For each span $S:\w\leftrightarrow \w'$ and $a\in A\w$ and $a'\in A\w'$ a set $AS(a,a')$ such that the set of quadruples $(S,a,a',p)$ with $p\in AS(a,a')$ is a predomain with continuous second and third projections and discrete ordering in the first component. 
\item For each parametric square $(u,u'):S\rightarrow S'$ a continuous 
function 
\[
A(u,u') : \Pi a\in A{\domL S}.\Pi a'\in A{\domR S}.AS(a,a') \rightarrow 
AS'(u.a,u'.a')
\]
\item For each parametric square $(\id,\id):S\rightarrow S'$ a continuous function 
\[
A(S,S') : \Pi a\in A{\domL S}.\Pi a'\in A{\domR S}.AS(a,a') \rightarrow 
AS'(a,a')
\]
\item Continuous functions of the following types, witnessing reflexivity, symmetry and transitivity in the ``heterogeneous sense'': 
\[\begin{array}{l}
\Pi\w.\Pi a\in A\w.Ar(\w)(a,a)\\
\Pi S.\Pi a\in A\domL S.\Pi a'\in A\domL{S}.AS(a,a')\rightarrow As(S)(a',a)\\
\Pi \w\ \w'\ \w''.\Pi S:\w\leftrightarrow \w'.\Pi S':\w'\leftrightarrow\w''.
\Pi a\in A\w.\Pi a'\in A\w'.\Pi a''\in A\w''.\\
\qquad AS(a,a')\times AS'(a',a'')\rightarrow At(S,S')(a,a'')\\
\end{array}\]
\item Continuous functions of the following types, witnessing the functorial laws:
\[\begin{array}{l}
\Pi \w.\Pi a\in A\w.Ar(\w)(a,\id.a)\\
\Pi \w\ \w_1\ \w_2.\Pi u:\w\rightarrow\w_1.\Pi v:\w_1\rightarrow\w_2.
Ar(\w_2)(v.u.a,(vu).a)
\end{array}\]
\end{enumerate}
\end{defi}

By way of motivating axiom (5), suppose that $S,S':\w\leftrightarrow\w'$ are isomorphic spans between $\w$
and $\w'$ in the sense that $(\id,\id):S\rightarrow S'$ where $t$ is the isomorphism
associated to $(\id,\id)$. If we have (an element of) $AS(a,a')$, then
axiom (4) suffices to establish $A(\id,\id)(a,a')\in
AS'(\id.a,\id.a')$. But what we really want is (an element of)
$AS'(a,a')$, which is just what axiom (5) provides in the shape of
$A(S,S')(a,a')$. Without axiom (5), one can get quite close: using
axioms (6) and (7) one can get (an element of)
$At(r(\w),t(S',s(r(\w))))(a,a')$, for example. But without explicitly
postulating axiom (5), as we do, or 
making extra assumptions on the $t(-,-)$ or
$\id.-$ operations, it seems impossible to reach $AS'(a,a')$. The following lemma is an instance of Axiom (5):

\begin{lem}
\label{lem:iso-spans}
  If $S \cong S'$ are isomorphic spans over $\w,\w'$, there is a continuous function of type:
   \[
     \Pi a\in A\w. \Pi a'\in A \w'. AS(a,a') \rightarrow AS'(a,a')
   \]
\end{lem}

\begin{lem}\label{luf}
Let $A$ be a parametric \svf. We have a continuous function of type:
\[
  \Pi a \in A\w. \Pi \w_1. \Pi u: \w \rightarrow \w_1.AS(a,u.a)
\]
where $S$ is $\spn{\w}{\id}{\w}{u}{\w_1}$. 
\end{lem}
\begin{proof}
We use the parametric square $(\id,u):S_0\rightarrow S$ where $S_0$ is $\spn{\w}{\id}{\w}{\id}{\w}$. 
\end{proof}
Every parametric \svf\ also is a plain \svf\ where we just define
$A\w(a,a')=Ar(\w)(a,a')$ and quotient the transition maps $Au$ by
pointwise $\sim$-equivalence.

But also every \svf\ $A$ can be extended to a parametric \svf:
first fix a particular choice of transition functions
$Au:A\w\rightarrow A\w'$ when $u:\w\rightarrow\w'$. Now define
$AS(a,a')=A\overline\w(x.a,x'.a')$ where $\overline \w$ is the apex of a
completion of $S$ to a minimal pullback and $x:\domL S\rightarrow
\overline \w$, $x':\domR S\rightarrow \overline \w$ are the
corresponding maps.

However, this correspondence is not one-to-one. For a concrete
counterexample, consider the following example which also lies at the
heart of the justification of ``Equivalence 12'' with parametric
functors.

\begin{exa}\label{myes}
The parametric \svf\ $[N{\Rightarrow}B]$ is defined by $[N{\Rightarrow}B]\w=2^\w$ (functions from $\w$ to $\{\mtrue,\mfalse\}$) and 
\[
 [N{\Rightarrow}B]S(f,f')=\begin{cases}
\{\star\}&\mbox{if $\forall n\in \lop S.f(S.u(n))=f'(S.u'(n))$}\\
\emptyset&\mbox{otherwise}
\end{cases}
\]
Now, let $S$ be $\spn{\{0\}}{}{\emptyset}{}{\emptyset}$ and put $f(x)=$``$x{=}0$'' and $f'(x)=\mfalse$. We have $[N{\Rightarrow}B]S(f,f')$, \ie, $[N{\Rightarrow}B]S(f,f')=\{\star\}$, thus $f$ and $f'$ are considered equal above span $S$. On the other hand, if we complete $S$ to a minimal pullback by $\{0\}\stackrel{1}{\rightarrow}\{0\}\stackrel{x'}{\leftarrow}\emptyset$ then $[N{\Rightarrow}B]r(\{0\})(f,x'.f)=\emptyset$, \ie, $f$ and $f'$ are not equal when regarded over the least common world, namely $\{0\}$. 
\end{exa}
\begin{defi}
A \emph{parametric natural transformation}, $f$, from parametric \svf\ $A$ to $B$ consists of two continuous functions 
\[\begin{array}{l}
 f_0 : \Pi\w.A\w\rightarrow B\w\\
 f_1 : \Pi S.\Pi a\in\domL S.\Pi a':\domR S.AS(a,a') \rightarrow BS(f_0{\domL S}(a),
f_0{\domR S}(a'))
\end{array}\]
As usual we refer both $f_0$ and $f_1$ as $f$.
Two parametric natural transformations $f,f':A\rightarrow B$ are identified if there is a continuous function of type
\[
\Pi \w.\Pi a\in A\w.Br(\w)(f\w(a),f'\w(a))
\]
\end{defi}
The identification of ``pointwise equal'' parametric natural transformations is meaningful as follows: 
\begin{lem}
Let $f$ and $f'$ be representatives of the same parametric natural transformation $A\rightarrow B$. There then is a continuous function of the following type:
\[
\Pi S.\Pi a\in\domL S.\Pi a'\in\domR S.AS(a,a')\rightarrow BS(f\domL S(a),f'\domR S(a'))
\]
\end{lem}
\begin{proof}
Given $S$, $a$, $a'$, and $p\in AS(a,a')$ we obtain $BS(f\domL{S}(a),f\domR{S}(a'))$ and we also obtain
$Br(\domR S)(f\domR S(a'),f\domR S(a'))$ since $f$ and $f'$ are pointwise equal. We conclude by transitivity and Lemma~\ref{lem:iso-spans}.
\end{proof}

\begin{lem}
If $f:A\rightarrow B$ is a parametric natural transformation then there are continuous functions of the following types
\[\begin{array}{l}
\Pi \w\ \w_1.\Pi u:\w\rightarrow \w_1.\Pi a\in A\w.Br(\w_1)(u.f\w(a),f\w_1(u.a))\\
\Pi\w.\Pi a\ a'\in A\w.Ar(\w)(a,a')\rightarrow Br(\w)(f\w(a),f\w(a'))
\end{array}\]
\end{lem}
\begin{proof}
Fix $u$, $a$ and $\w$. Lemma~\ref{luf} furnishes an element of $AS(a,u.a)$ where 
$S$ is $\spn{\w}{1}{\w}{u}{\w_1}$. Since $f$ is a parametric natural transformation, we then get an element of $BS(f\w(a),f\w_1(u.a))$. We then get the desired element of 
$Br(\w_1)(u.f\w(a),f\w_1(u.a))$ by applying parametricity of  $B$ to the parametric square $(u,1):S\rightarrow r(\w_1)$. 
\end{proof}
\begin{thm}
The parametric \svfs\ with parametric natural transformations form a cartesian closed category with fixpoint operator obeying the laws from Theorems~\ref{tfixit} \& \ref{fixte}. There is a strong monad $T$ on this category where 
\[\begin{array}{l}
TA\w = \{(\w_1,a)\mid \w\subseteq \w_1\wedge a\in A\w_1\}_\bot\\
TAS((\w_1,a), (\w_1',a')) = \{(S_1:\w_1\leftrightarrow\w_1',p) | 
(\w{\hookrightarrow}\w_1,\w'{\hookrightarrow}\w_1'):S\rightarrow S_1 \wedge p\in AS'(a,a')\}
\end{array}\]
\end{thm}
\begin{proof}[Proof (sketch)]
The interesting bit is the proof of transitivity for the monad, which
seems to rely in an essential manner on proof relevance. Suppose that $S:\w\leftrightarrow \w'$ and $S':\w'\leftrightarrow\w''$ and that $(\w_1,a)\in TA\w$ and $(\w_1',a')\in TA\w'$ and $(\w_1'',a'')\in TA\w''$. Furthermore, suppose that 
$(S_1,p)\in TAS((\w_1,a),(\w_1',a'))$ and $(S_1',p')\in TAS'((\w_1',a'),(\w_1'',a''))$. 

Now, by definition, we have $S_1:\w_1\leftrightarrow \w_1'$ and $S_1':\w_1'\leftrightarrow \w_1''$ and also $p\in AS_1(a,a')$ and $p'\in AS_1'(a',a'')$. We thus obtain (an element of) $At(S_1,S_1')(a,a'')$ and this, together with $t(S_1,S_1')$ furnishes the required proof.
\end{proof}
\begin{rem}
Notice that if the extensions $\w_1$ were existentially quantified, as
in more traditional non-proof-relevant formulations of Kripke logical
relations (e.g.~that of Stark
\cite[Section 4.1]{Stark94thesis}), then
the transitivity construction in the above proof would not have been
possible because we would have no guarantee that the existential
witnesses used in the two assumptions are the same.
\end{rem}

 \subsection{Private Name Equation}
We now return to our motivating equivalence, illustrating that a
function value may encapsulate a freshly generated name without
revealing it to the context:
\[
 (\letin{n}{\New} \lambda x. (x = n) ) = (\lambda x. \mfalse).
\]
Writing $\ccval$ and $\ccval'$ for the LHS and RHS of the above,
the equivalence proof is based on the following diagram:
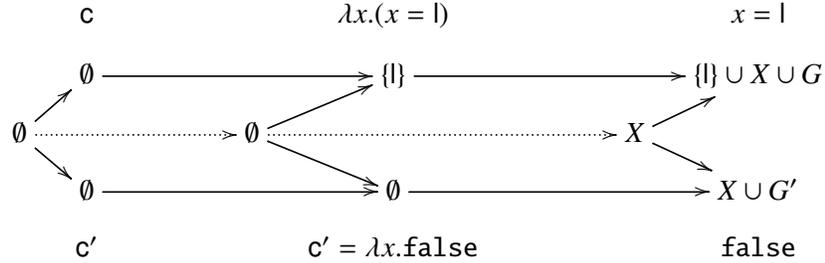
\begin{displaymath}
\vcenter{\xymatrix@C=1pc@R=0.6pc{
& \cval &&&& \lambda x. (x = \cloc) &&&& x = \cloc\\
& \emptyset  \ar[rrrr] & & & & \{\cloc\} \ar[rrrr] & & & & \{\cloc\} \cup X \cup G  \\
\emptyset \ar@{.>}[rrrr] \ar[ur]\ar[dr] &&  & & \emptyset\ar@{.>}[rrrr] \ar[ru] \ar[rd] & & & & X \ar[ur] \ar[dr] & &   \\
& \emptyset  \ar[rrrr] & & & & \emptyset  \ar[rrrr] & & & & X \cup G' \\
& \cval' &&&& \cval'=\lambda x. \mfalse &&&& \mfalse
}}
\end{displaymath}
We show that $\cval$ and $\cval'$ are equivalent in the trivial span to the left. For the generation 
of the fresh value in $\cval$, we choose the extension of the worlds with the fresh value $\cloc$, the second span 
 shown in the diagram. Now it remains to prove that $\lambda x. x = \cloc$ and $\cval'$ are equivalent above the latter. 
This means that for any extension of worlds, $x = \cloc$ and $\mfalse$ should be related. 
Consider the extension of worlds in the right-most span in the diagram above.
The names in $X$ denote the 
common names, while $G$ and $G'$ the spurious names created. Notice that $\cloc$ is not in the low point of the third span
because the squares with vertices $\emptyset, X, \{\cloc\} \cup X \cup G, \{\cloc\}$ and $\emptyset, X,  X \cup G', \emptyset$ 
are pullbacks as by Definition~\ref{parames}. 
Thus, the value of $x$ cannot be $\cloc$ and $x = \cloc$ is indeed equal to $\mfalse$.

\section{Discussion}
We have introduced proof-relevant logical relations and shown how they
may be used to model and reason about simple equivalences in a
higher-order language with recursion and name generation. A key
innovation compared with previous functor category models is the use
of functors valued in setoids (which are here also built on
predomains), rather than plain sets. One payoff is that we can work
with a direct style model rather than one based on continuations
(which, in the absence of control operators in the language, is less
abstract).

The technical machinery used here is not \emph{entirely} trivial, and
the reader might be forgiven for thinking it slightly excessive for
such a simple language and rudimentary equations. However, our aim has
not been to present impressive new equivalences, \squelch{(though less
  trivial ones do hold, even in this simple case),} but rather to
present an accessible account of how the idea of proof relevant
logical relations works in a simple setting. The companion paper
\cite{benton14popl} gives significantly more advanced examples of
applying the construction to reason about equivalences justified by
abstract semantic notions of effects and separation, but the way in
which setoids are used is there potentially obscured by the details
of, for example, much more sophisticated categories of worlds.  Our
hope is that this account will bring the idea to a wider audience,
make the more advanced applications more accessible, and inspire
others to investigate the construction in their own work.

Thanks to Andrew Kennedy for numerous discussions, to the
referee who first suggested that we write up the details of how
proof-relevance applies to pure name generation, and to the referees
of the present paper for their many helpful suggestions.

\begin{thebibliography}{AGM{\etalchar{+}}04}

\bibitem[AGM{\etalchar{+}}04]{AbramskyEtal04lics}
S.~Abramsky, D.~R. Ghica, A.~S. Murawski, C.-H.~L. Ong, and I.~D.~B. Stark.
\newblock Nominal games and full abstraction for the nu-calculus.
\newblock In {\em Proc. 19th Annual IEEE Symposium on Logic in Computer Science
  (LICS '04)}. IEEE Computer Society, 2004.

\bibitem[BB06]{bohrbirkedal}
N.~Bohr and L.~Birkedal.
\newblock Relational reasoning for recursive types and references.
\newblock In {\em Proc. Fourth Asian Symposium on Programming Languages and
  Systems (APLAS '06)}, volume 4279 of {\em LNCS}. Springer, 2006.

\bibitem[BCP03]{DBLP:journals/jfp/BartheCP03}
Gilles Barthe, Venanzio Capretta, and Olivier Pons.
\newblock Setoids in type theory.
\newblock {\em J. Funct. Program.}, 13(2):261--293, 2003.

\bibitem[BCRS98]{DBLP:conf/lics/BirkedalCRS98}
Lars Birkedal, Aurelio Carboni, Giuseppe Rosolini, and Dana~S. Scott.
\newblock Type theory via exact categories.
\newblock In {\em Proc. 13th Annual IEEE Symposium on Logic in Computer Science
  (LICS '98)}. IEEE Computer Society, 1998.

\bibitem[B{\'E}93]{DBLP:series/eatcs/BloomE93}
Stephen~L. Bloom and Zolt{\'a}n {\'E}sik.
\newblock {\em Iteration Theories - The Equational Logic of Iterative
  Processes}.
\newblock EATCS Monographs on Theoretical Computer Science. Springer, 1993.

\bibitem[BHN14]{benton14popl}
Nick Benton, Martin Hofmann, and Vivek Nigam.
\newblock Abstract effects and proof-relevant logical relations.
\newblock In {\em Proc. 41st ACM SIGPLAN-SIGACT Symposium on Principles of
  Programming Languages (POPL '14)}. ACM, 2014.

\bibitem[BK13]{bentonkoutavas}
N.~Benton and V.~Koutavas.
\newblock A mechanized bisimulation for the nu-calculus.
\newblock {\em Higher-Order and Symbolic Computation}, 2013.
\newblock To appear.

\bibitem[BKBH07]{DBLP:conf/ppdp/BentonKBH07}
Nick Benton, Andrew Kennedy, Lennart Beringer, and Martin Hofmann.
\newblock Relational semantics for effect-based program transformations with
  dynamic allocation.
\newblock In {\em Proc. Ninth International ACM SIGPLAN Symposium on Principles
  and Practice of Declarative Programming (PPDP '07)}. ACM, 2007.

\bibitem[BKHB06]{DBLP:conf/aplas/BentonKHB06}
Nick Benton, Andrew Kennedy, Martin Hofmann, and Lennart Beringer.
\newblock Reading, writing and relations: Towards extensional semantics for
  effect analyses.
\newblock In {\em Proc. Fourth Asian Symposium on Programming Languages and
  Systems (APLAS '06)}, volume 4279 of {\em LNCS}. Springer, 2006.

\bibitem[BL05]{DBLP:conf/tlca/BentonL05}
Nick Benton and Benjamin Leperchey.
\newblock Relational reasoning in a nominal semantics for storage.
\newblock In {\em Proc. Seventh International Conference on Typed Lambda
  Calculi and Applications (TLCA '05)}, volume 3461 of {\em LNCS}. Springer,
  2005.

\bibitem[CFS87]{DBLP:conf/mfps/CarboniFS87}
Aurelio Carboni, Peter~J. Freyd, and Andre Scedrov.
\newblock A categorical approach to realizability and polymorphic types.
\newblock In {\em Proc. Third Workshop on Mathematical Foundations of
  Programming Language Semantics (MFPS '87)}, volume 298 of {\em LNCS}.
  Springer, 1987.

\bibitem[GP02]{DBLP:journals/fac/GabbayP02}
Murdoch Gabbay and Andrew~M. Pitts.
\newblock A new approach to abstract syntax with variable binding.
\newblock {\em Formal Asp. Comput.}, 13(3-5):341--363, 2002.

\bibitem[nLa]{nlab}
nLab.
\newblock cartesian closed category.
\newblock
  \url{https://ncatlab.org/nlab/show/cartesian+closed+category#exponentials_of_cartesian_closed_categories}.

\bibitem[PS93]{DBLP:conf/mfcs/PittsS93}
Andrew~M. Pitts and Ian D.~B. Stark.
\newblock Observable properties of higher-order functions that dynamically
  create local names, or what's new?
\newblock In {\em Proc. 18th International Symposium on Mathematical
  Foundations of Computer Science (MFCS '93)}, volume 711 of {\em LNCS}.
  Springer, 1993.

\bibitem[PS98]{pitts98high}
Andrew Pitts and Ian Stark.
\newblock Operational reasoning for functions with local state.
\newblock In {\em Higher Order Operational Techniques in Semantics}, pages
  227--273. Cambridge University Press, 1998.

\bibitem[Shi04]{shinwell04phd}
Mark~R. Shinwell.
\newblock {\em The Fresh Approach: functional programming with names and
  binders}.
\newblock PhD thesis, University of Cambridge, 2004.

\bibitem[SP00]{DBLP:conf/lics/SimpsonP00}
Alex~K. Simpson and Gordon~D. Plotkin.
\newblock Complete axioms for categorical fixed-point operators.
\newblock In {\em LICS}, pages 30--41. IEEE Computer Society, 2000.

\bibitem[SP05]{shinwell05tcs}
Mark~R. Shinwell and Andrew~M. Pitts.
\newblock On a monadic semantics for freshness.
\newblock {\em Theor. Comput. Sci.}, 342(1):28--55, 2005.

\bibitem[SPG03]{shinwell03icfp}
Mark~R. Shinwell, Andrew~M. Pitts, and Murdoch~J. Gabbay.
\newblock {FreshML}: Programming with binders made simple.
\newblock In {\em Proc. Eighth ACM SIGPLAN International Conference on
  Functional programming (ICFP '03)}. ACM, 2003.

\bibitem[Sta94]{Stark94thesis}
I.~D.~B. Stark.
\newblock {\em Names and Higher-Order Functions}.
\newblock PhD thesis, University of Cambridge, Cambridge, UK, December 1994.
\newblock Also published as Technical Report~363, University of Cambridge
  Computer Laboratory.

\bibitem[Tze12]{tzevelekos}
N.~Tzevelekos.
\newblock Program equivalence in a simple language with state.
\newblock {\em Computer Languages, Systems and Structures}, 38(2), 2012.

\end{thebibliography}

\newcommand{\etalchar}[1]{$^{#1}$}

\end{document}